\documentclass[
reprint, 
superscriptaddress, 
aps, 
prl, 
]{revtex4-2}

\pdfoutput=1

\usepackage{graphicx} 
\usepackage{dcolumn} 
\usepackage{bm} 

\usepackage{amsthm,amssymb, physics}
\usepackage{xcolor}
\usepackage{soul}
\usepackage{url}
\usepackage{enumitem}
\usepackage{appendix}
\usepackage{mathtools}
\usepackage{quantikz}
\usepackage{tikz}
\usepackage{makecell}
\usepackage[normalem]{ulem}
\usepackage{hyperref}
\usepackage[hyphenbreaks]{breakurl}

\newcommand{\mcl}[1]{\ensuremath{\mathcal{#1}}}

\newcommand{\bP}{\mathbb{P}}

\newcommand{\Ilamb}[1][\lambda]{\ensuremath{I_{#1}}}
\newcommand{\oprep}[1][\rho]{\ensuremath{\mu_{#1}}}

\newcommand{\omeas}[1]{\ensuremath{\xi_{#1}}}

\newcommand{\DChannel}[1]{\ensuremath{\mathcal{M}^{D_#1}}}

\newcommand{\OnticMeasureBProb}[1]{\omeas{\projB{#1}}}
\newcommand{\OnticMeasureAProb}[1]{\omeas{\projA{#1}}}
\newcommand{\OnticXWeakMeasureAProb}[1]{\omeas{\projA{#1}}^{\mathrm{WX}}}
\newcommand{\OnticYWeakMeasureAProb}[1]{\omeas{\projA{#1}}^{\mathrm{WY}}}
\newcommand{\OnticDTransformation}[1]{\ensuremath{\Gamma_{D_{#1}}}}

\newcommand{\wvrkraus}[1]{\ensuremath{N_{#1}}}
\newcommand{\wvikraus}[1]{\ensuremath{M_{#1}}}
\newcommand{\Doperator}[1]{\ensuremath{D}_{#1}}
\newcommand{\ProbOfBOnExot}[2]{\ensuremath{p^{#1}_{#2}}}
\newcommand{\ProbOfBOnDExotD}[2]{\ensuremath{p^{D_{#1}}_{#2}}}

\newcommand{\UndetermindedStochasticMatrix}[1]{\phi_{#1}}

\newcommand{\projA}[1]{\ensuremath{P}_{#1}}
\newcommand{\projB}[1]{\ensuremath{\Pi}_{#1}}
\newcommand{\KD}[1]{\ensuremath{Q}_{#1}}
\newcommand{\ArbHermOperator}{\ensuremath{H}}
\newcommand{\exot}{\ensuremath{\rho_\star}}
\newcommand{\KDexot}{\ensuremath{\mcl{E}^{\mathrm{exot}}_{\mathrm{KD+}}}}
\newcommand{\conv}[1]{\mathrm{conv}\left(#1\right)}

\newcommand{\EcalKDCpu}{{\mcl E}_{\mathrm{KD+}}^{\mathrm{pure}}}

\newcommand{\OutcomePMF}[2]{\ensuremath{f^{(#1)}_{#2}}}

\newcommand{\NPr}{{\mcl N}{({\rho})}}

\newcommand{\ProbabilityOfMeasurement}{\ensuremath{p_\mathrm{m}}}
\newcommand{\ProbabilityOfDisturbance}{\ensuremath{p_\mathrm{d}}}

\newtheorem{Theorem}{Theorem}
\newtheorem{Lemma}[Theorem]{Lemma}

\begin{document}

\title{Contextuality Can be Verified with Noncontextual Experiments}

\author{Jonathan J. Thio}
\affiliation{Cavendish Lab., Department of Physics, Univ. of Cambridge, Cambridge CB3 0HE, United Kingdom}

\author{Wilfred Salmon}
\affiliation{DAMTP, Centre for Mathematical Sciences, Univ. of Cambridge, CB30WA, United Kingdom}
\affiliation{Hitachi Cambridge Laboratory, J. J. Thomson Avenue, Cambridge CB3 0HE, United Kingdom}

\author{Crispin H. W. Barnes}
\affiliation{Cavendish Lab., Department of Physics, Univ. of Cambridge, Cambridge CB3 0HE, United Kingdom}

\author{Stephan De Bièvre}
\affiliation{Univ. Lille, CNRS, Inria, UMR 8524, Laboratoire Paul Painlevé, F-59000 Lille, France}

\author{David R. M. Arvidsson-Shukur}
\affiliation{Hitachi Cambridge Laboratory, J. J. Thomson Avenue, Cambridge CB3 0HE, United Kingdom}
\date{\today} 

\begin{abstract}
We uncover new features of generalized contextuality by connecting it to the Kirkwood-Dirac (KD) quasiprobability distribution. Quantum states can be represented by KD distributions, which take values in the complex unit disc. Only for ``KD-positive'' states are the KD distributions joint probability distributions.  A KD distribution can be measured by a series of weak and projective measurements. We design such an experiment and show that it is contextual iff the underlying state is not KD-positive.  We analyze this connection with respect to mixed KD-positive states that cannot be decomposed as convex combinations of pure KD-positive states. Our result is the construction of a noncontextual experiment that enables an experimenter to verify contextuality.
\end{abstract}

\maketitle

\textit{Introduction:---}
Quantum physics is fundamentally different from classical physics. However, pinpointing exactly what is nonclassical is famously difficult. While certain experiments necessitate quantum theory for a full description, most can be described with classical models. This begs the question: Where does the boundary between classical and nonclassical experiments lie? Investigating this quantum--classical boundary has proven fruitful for computational \cite{Gottesman2004, Weedbrook2012, Veitch_2012, Veitch_2014, Raussendorf20} and metrological \cite{Giovannetti2006, Giovanetti11, Maccone2013, Arvidsson-Shukur2020, lupu2022negative} applications. Here, we study the boundary from a more foundational perspective.

A trending and rigorous notion of nonclassicality is generalized contextuality \cite{spekkens2005contextuality, spekkens2008negativity, Quanta22, Mazurek2016, Schmid2018, Lostaglio2020}. To determine, within this notion, if an experiment is classical, one describes it with a hidden-variable model. In such a model, one considers a set $\Lambda$ of hidden variables. A preparation $P$ is described by a probability distribution $\mu_P(\lambda)$, where $\lambda \in \Lambda$. A transformation $T$ is described by a transition matrix between hidden variables, $\Gamma_{T}(\lambda'|\lambda)$. A measurement $M$ is described by a probability distribution on the outcome set conditioned on the hidden variable, $\xi_{M}(k|\lambda)$, where $k$ is a specific outcome.  The probability of obtaining outcome $k$ given the $3$-tuple ($P$, $T$, $M$), is  
\begin{equation}\label{eq:HowHVMMAkePredictions}
    \mathbb{P}(k|P,T,M) = \int_{\Lambda \times \Lambda} d\lambda d\lambda' \mu_{P}(\lambda) \Gamma_{T}(\lambda'|\lambda) \xi_M (k|\lambda').
\end{equation}
An experiment is composed of a set of such $3$-tuples. The hidden-variable model can then be tuned to match predictions of quantum theory or experimental data. This can always be done \cite{spekkens2005contextuality}. However, the hidden-variable model may have nonclassical features, for example, nonlocality \cite{Bell1964}. We capture nonclassicality using the notion of generalized contextuality \cite{spekkens2005contextuality}. We say that a hidden-variable model is noncontextual iff it assigns the same probability distribution to experimentally indistinguishable procedures. We say that an experiment is contextual iff there does not exist a noncontextual hidden-variable model describing it. Nevertheless, in this article, we construct an experiment [Fig. \ref{fig:AliceAndBobExperiment}], the contextuality of which can be verified with noncontextual procedures. This experiment relies on measurements of a Kirkwood-Dirac (KD) distribution.

\begin{figure}
    \centering
    \begin{tikzpicture}[]

    \draw[fill=blue!20, thick] (-2,0.2) rectangle (1.5,1.7);
    \node at (-1.3,0.7) {$(\psi_j)_{j=1}^N$};
    \node[draw, fill=white, text width=1.3cm, align=center] at (0.5,0.7) {Secret order};
    \node at (-0.25,1.4) {Alice};
    \draw[very thick, ->] (-0.75,0.7) -- (-0.25,0.7);
    
    \draw[fill=yellow!20, thick] (1.7,0.2) rectangle (6.0,1.7);
    \draw[very thick, ->] (1.25,0.7) -- (2.2,0.7);
    \node at (3.65,1.4) {Bob};
    \node at (2.4,0.7){$\exot$};
    \draw[very thick, ->] (2.6,0.7) -- (3.1,0.7);
    \node[draw, fill=white, text width=2.6cm, align=center] at (4.5,0.7) {Protocols 1 - 6};
    \draw[thick, ->] (4.5,0.35) -- (4.5,-0.2);

    \draw[rounded corners=10pt, dashed] (-2,-0.6) rectangle (6,-0);
    \node at (4.5, -0.3) {Data};
    \node at (2,-0.3) {Public Channel};
    \draw[thick, ->] (-0.5,-0.2) -- (-0.5,0.35);
    
    \end{tikzpicture}
    \caption{\textbf{Noncontextual Experiment that Signals Contextuality.} Alice sends the sequence $(\psi_j)_{j=1}^N$ of pure states  to Bob without disclosing the order. Bob effectively receives a mixed state $\exot = \frac{1}{N} \sum_{j=1}^N \psi_j$. Bob performs  Procedures 1 to 6 (outlined in the main text). His experiment is noncontextual. Nonetheless, he can verify that Alice's experiment is contextual.}
    \label{fig:AliceAndBobExperiment}
\end{figure}
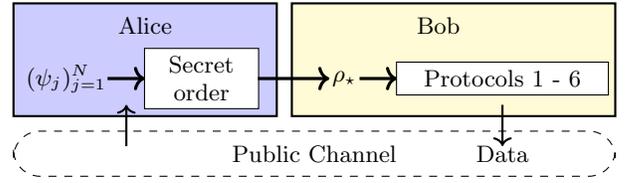
\textit{KD Distributions:---} The KD distributions \cite{Kirkwood33,Dirac45,arvidssonshukur2024properties,lostaglio2023kirkwood} are a family of quasiprobability distributions that have recently found numerous applications in the field of quantum information processing. We consider a Hilbert space of dimension $d < \infty$. The KD distribution $Q$ maps a quantum state $\rho$ to a corresponding quasiprobability distribution based on two nondegenerate observables $A = \sum_j a_j \projA{j}$ and $B = \sum_k b_k \projB{k}$. Here, $\{ a_j \}$ and $\{ b_k \}$ are the eigenvalues of $A$ and $B$, respectively. $\projA{j}$ and $\projB{k}$ are the rank-$1$ projectors onto the corresponding eigenspaces. The KD distribution of a state $\rho$ is 
\begin{equation}
    \KD{j,k}(\rho) \coloneqq \Tr (\projB{k}\projA{j}\rho).
\end{equation}
Throughout this work, we assume that $\projA{j} \neq \projB{k}$ for all $j$ and $k$. Then, $Q$ is invertible: Knowledge of $Q(\rho)$ enables an informationally complete reconstruction of $\rho$ \cite{arvidssonshukur2024properties}. 

The KD distribution is not a proper probability distribution as its entries may lie in the complex unit disc. However, for certain quantum states, all the entries of the KD distribution lie in the interval $[0,1]$, in which case the KD distribution is a probability distribution. We say that a quantum state is KD-positive iff this is the case. Otherwise, we call it KD-nonpositive. One can quantify how nonpositive a KD distribution is by its nonpositivity:
\begin{equation}
\NPr = -1 + \sum_{j,k}  \left| \KD{j,k}(\rho) \right| \geq 0,
\end{equation}
\noindent where $\rho$ is KD-positive iff $\NPr=0$ \cite{YungHalp19, arvidssonshukur2024properties}. For example, if $\rho$ equals an eigenstate of $A$ or $B$, then $\NPr=0$. However, for pure states, if $\rho$ has nonzero overlap with sufficiently many of the eigenstates of $A$ and $B$, then $\rho$ is KD-nonpositive: $\NPr > 0$ \cite{ArvidssonShukur2021, debievre2021}. It was recently shown, that for almost all choices of $A$ and $B$, all KD-positive states are convex mixtures of $A$ and $B$'s basis states \cite{ADBLSLT24}. However, for certain $A$ and $B$, there exist mixed states that are KD-positive, but that cannot be written as convex combinations of pure KD-positive states \cite{langrenez2023characterizing}. We call such states `exotic', and we denote the set containing them by $\KDexot$. Exotic states will play the central role in what follows. 

\textit{Measurements of $\KD{}(\rho)$:---}To measure $\KD{} (\rho)$, one can perform a series of protocols on $\rho$, involving the projective measurement of $\projA{j}$ and $\projB{k}$, and the so-called `weak measurement' \cite{Aharanov88,Duck89,Jozsa2007,Dressel14} of $\projA{j}$. The weak measurement involves coupling weakly the $A$ observable of $\rho$ to a qubit ancilla, which later is measured in the $X$ or $Y$ eigenbasis. The strength of the coupling is given by the parameter $\epsilon \ll 1$. The exact implementation of the weak measurement is given in Note I of the Supplementary Material. For our purposes, it suffices to note that the $X$- and $Y$-type weak measurements of $\projA{j}$ are represented by the Kraus operators 
\begin{align}\label{eq:KrausOperatorsXPointer}
    \wvrkraus{x,j} & = \frac{1}{\sqrt{2}}(\cos \epsilon I + x \sin \epsilon \Doperator{j}), \\
\label{eq:KrausOperatorYPointer}
        \wvikraus{y,j}  & = \frac{1}{\sqrt{2}}(\cos \epsilon I -i y \sin \epsilon \Doperator{j}), 
\end{align}
respectively. Here, $\Doperator{j} \coloneqq 2\projA{j} - I$ and $x,y \in \{\pm 1 \}$ are the outcomes of the two weak measurements. When $\epsilon = \pi/2$, then $\wvrkraus{x,j}$ and $\wvikraus{x,j}$ are projective. When $\epsilon \ll 1$, the weak measurements convey little information about the observable.

To measure the KD distribution, we consider the following six protocols [Fig. \ref{fig:protocols}] on a system in the initial state $\rho$. Protocols $1$ to $3$ measure the `weak values' \cite{Dressel14,Aharanov88,Duck89,Jozsa2007} of $\projA{j}$. Protocols $4$ through $6$ are designed to establish the connection to contextuality, outlined below.

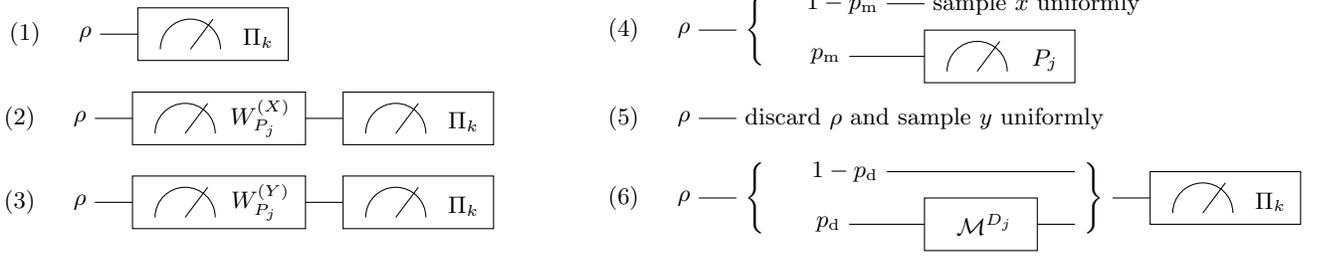
\begin{figure*}\label{fig:procedures}
    \centering
    \hspace{-40pt}
    \begin{tabular}{@{}l@{} @{}l@{}}
        \begin{tabular}{c}
            \begin{tikzpicture}
                \draw (-1.5,0.5) node {$(1)$};
                \draw (-0.5,0.5) node[left] {$\rho$} -- (0,0.5);
                \draw (0,0.15) rectangle (2,0.85);
                \draw (1.1,0.3) arc[start angle=0, end angle=180, radius=0.4];
                \draw (0.7,0.3) -- (1.0,0.7);
                \node at (1.6,0.45) {$\projB{k}$};
            \end{tikzpicture}  \\
        \end{tabular}
        &
        \begin{tabular}{@{}l@{}}
            \begin{tikzpicture}\hspace{40pt}
                \draw (-1.5,0.5) node {$(4)$};
                \draw (-0.5,0.5) node[left] {$\rho$} -- (0,0.5);
                \node[right] at (0,0.5) {$\Biggl\{$};
                \draw (2,0.85) node[left] {$1-\ProbabilityOfMeasurement$} -- (2.5,0.85) node[right]{sample $x$ uniformly};
                \draw (1.5,0.15) node[left] {$\ProbabilityOfMeasurement$} -- (2.5,0.15);
                \draw (2.5,-0.20) rectangle (4.5,0.5);
                \draw (3.6,-0.05) arc[start angle=0, end angle=180, radius=0.4];
                \draw (3.2,-0.05) -- (3.5,0.35);
                \node at (4.1,0.1) {$\projA{j}$};
            \end{tikzpicture}\\ 
        \end{tabular}\\
        \begin{tabular}{@{}l@{}}
            \begin{tikzpicture}
                \draw (-1.5,0.5) node {$(2)$};
                \draw (-0.5,0.5) node[left] {$\rho$} -- (0,0.5);
                \draw (0,0.15) rectangle (2.3,0.85);
                \draw (1.1,0.3) arc[start angle=0, end angle=180, radius=0.4];
                \draw (0.7,0.3) -- (1.0,0.7);
                \node at (1.7,0.5) {$W^{(X)}_{\projA{j}}$};
                \draw (2.3,0.5) -- (2.8,0.5);
                \draw (2.8,0.15) rectangle (4.8,0.85);
                \draw (3.9,0.3) arc[start angle=0, end angle=180, radius=0.4];
                \draw (3.5,0.3) -- (3.8,0.7);
                \node at (4.4,0.45) {$\projB{k}$};
            \end{tikzpicture} \\ 
        \end{tabular}
        &
        \begin{tabular}{@{}l@{}}
            \begin{tikzpicture}\hspace{40pt}
                \draw (-1.5,0.5) node {$(5)$};
                \draw (-0.5,0.5) node[left] {$\rho$} -- (0,0.5) node[right]{discard $\rho$ and sample $y$ uniformly};
            \end{tikzpicture} \\ 
        \end{tabular}
        \\
        \begin{tabular}{@{}l@{}}
            \begin{tikzpicture}
                \draw (-1.5,0.5) node {$(3)$};
                \draw (-0.5,0.5) node[left] {$\rho$} -- (0,0.5);
                \draw (0,0.15) rectangle (2.3,0.85);
                \draw (1.1,0.3) arc[start angle=0, end angle=180, radius=0.4];
                \draw (0.7,0.3) -- (1.0,0.7);
                \node at (1.7,0.5) {$W^{(Y)}_{\projA{j}}$};
                \draw (2.3,0.5) -- (2.8,0.5);
                \draw (2.8,0.15) rectangle (4.8,0.85);
                \draw (3.9,0.3) arc[start angle=0, end angle=180, radius=0.4];
                \draw (3.5,0.3) -- (3.8,0.7);
                \node at (4.4,0.45) {$\projB{k}$};
            \end{tikzpicture} \\ 
        \end{tabular}
        &
        \begin{tabular}{@{}l@{}}
            \begin{tikzpicture}\hspace{40pt}
                \draw (-1.5,0.5) node {$(6)$};
                \draw (-0.5,0.5) node[left] {$\rho$} -- (0,0.5);
                \node[right] at (0,0.5) {$\Biggl\{$};
                \draw (2,0.85) node[left] {$1-\ProbabilityOfDisturbance$} -- (4.5,0.85);
                \draw (1.5,0.15) node[left] {$\ProbabilityOfDisturbance$} -- (2.5,0.15);
                \draw (2.5,-0.2) rectangle (4,0.5);
                \node at (3.3,0.15) {$\DChannel{j}$};
                \draw (4,0.15) -- (4.5,0.15);
                \node[right] at (4.5,0.5) {$\Biggl\}$};
                \draw (5,0.5) -- (5.5,0.5);
                \draw (5.5,0.15) rectangle (7.5,0.85);
                \draw (6.6,0.3) arc[start angle=0, end angle=180, radius=0.4];
                \draw (6.2,0.3) -- (6.5,0.7);
                \node at (7.1,0.45) {$\projB{k}$};
            \end{tikzpicture} \\ 
        \end{tabular}
    \end{tabular}
    \caption{\textbf{Experimental Protocols.} Upon receiving a state $\rho$ from Alice, Bob randomly implements one of six protocols [see Fig. \ref{fig:AliceAndBobExperiment}]. The boxes with $\projA{j}$ and $\projB{k}$ represent projective measurements. The boxes labeled by $W^{(X)}_{\projA{j}}$ and  $W^{(Y)}_{\projA{j}}$ represent $X$-type and $Y$-type weak measurements of $\projA{j}$, respectively. The quantum theoretical predictions for these protocols and their contextual implications are given in the main body of the article.}
    \label{fig:protocols}
\end{figure*}

\begin{enumerate}
    \item Measure $\projB{k}$ on the system (returning outcome $z$).
    \item First, perform an $X$-type weak measurement of $\projA{j}$ (returning outcome $x$) and then measure $\projB{k}$ (returning outcome $z$).
    \item First perform an $Y$-type weak measurement of $\projA{j}$ (returning outcome $y$) and then measure $\projB{k}$ (returning outcome $z$).
    \item Either, with probability $\ProbabilityOfMeasurement \coloneqq \sin 2\epsilon$ , measure $\projA{j}$ (returning outcome $x$) or, with probability $1-\ProbabilityOfMeasurement$, return  $x\in\{\pm 1\}$ uniformly at random.
    \item Discard the system and  return as output $y\in\{\pm 1\}$ uniformly at random.
    \item Either apply the quantum channel $\DChannel{j}(\rho) \coloneqq \Doperator{j}\rho\Doperator{j}^{\dagger}$ with probability $\ProbabilityOfDisturbance \coloneqq \sin^2 \epsilon$ or apply the identity channel with probability $1-\ProbabilityOfDisturbance$, and finally measure $\projB{k}$ (returning outcome $z$). 
\end{enumerate}
These six protocols include repetition over all $j$ and $k$.
We denote the outcome-probability distributions for these protocols by $\OutcomePMF{1}{k}(z)$, $\OutcomePMF{1}{j,k}(x,z)$, $\OutcomePMF{3}{j,k}(y,z)$, $\OutcomePMF{4}{j}(x)$, $\OutcomePMF{5}{}(y)$ and $\OutcomePMF{6}{j,k}(z)$. Here $x,y,z \in \{\pm 1\}$.  The quantum-theoretical predictions for the outcome-probability distributions can be calculated using Eqs. \eqref{eq:KrausOperatorsXPointer} and \eqref{eq:KrausOperatorYPointer}. We introduce the shorthand $\projB{k}^z = \projB{k}$ if $z = +1$ and $\projB{k}^z = I - \projB{k}$ if $z = -1$. We find that
\begin{align}
    \label{eq:f1}
    \OutcomePMF{1}{k}(z)  =&  \Tr (\projB{k}^z \rho), \\
    \label{eq:f2}
\OutcomePMF{2}{j,k}(x,z) = &  \left(\frac{1}{2} - x\epsilon\right)\Tr (\projB{k}^z \rho) + x\epsilon(1-z)\Tr(\projA{j}\rho) \nonumber \\
& + 2xz\epsilon\Re(\KD{j,k}(\rho)) + O(\epsilon^2), \\
\label{eq:f3}
    \OutcomePMF{3}{j,k}(y,z) =& \frac{1}{2}\Tr(\projB{j}^z\rho) + 2yz\epsilon \Im \KD{j,k}(\rho) + O(\epsilon^2).
\end{align}
Quantum theory allows us to express the outcome-probability distributions of Protocols \eqref{eq:f4} - \eqref{eq:f6} in terms of the outcome-probability distributions of protocols \eqref{eq:f1} - \eqref{eq:f3}:
\begin{align}
\label{eq:f4}
\OutcomePMF{4}{j}(x) =& \sum_z \OutcomePMF{2}{j,k}(x,z)\\
\label{eq:f5}
\OutcomePMF{5}{}(y) =& \sum_z \OutcomePMF{3}{j,k}(y,z)\\
\label{eq:f6}
\OutcomePMF{6}{j,k}(z) =& \sum_x \OutcomePMF{2}{j,k}(x,z) = \sum_y \OutcomePMF{3}{j,k}(y,z)
\end{align}
(Detailed calculations are given in Note I of the Supplementary Material.)
When $z = +1$, then $x\epsilon(1-z)\Tr(\projA{j}\rho) = 0$, and the real and imaginary parts of the KD distribution can be deduced from Eqs. \eqref{eq:f1}, \eqref{eq:f2} and \eqref{eq:f3}.

\textit{KD-nonpositivity as a faithful witness of contextuality:---}Equations \eqref{eq:f4} - \eqref{eq:f6} can be used to establish noncontextuality constraints. The probability distribution for the fifth protocol can be obtained by summing over the $z$ variable in the probability distribution for the third protocol [Eq. \eqref{eq:f5}]. This implies that, within quantum theory, one cannot distinguish: (i) a $Y$-type weak measurement of $\projA{j}$ followed by a discarding of the system, and (ii) entirely ignoring the system and picking an outcome uniformly at random. Thus, in any noncontextual hidden-variable model, these two procedures are represented by the same probability distribution:
\begin{equation}
    \OnticYWeakMeasureAProb{j}(y|\lambda) = \frac{1}{2}. \label{eq:NonContextualityConstraint1}
\end{equation}
\noindent Here, $\OnticYWeakMeasureAProb{j}(y|\lambda)$ is the probability distribution representing a $Y$-type weak measurement of $\projA{j}$ followed by a discarding the system. Similarly, Eqs. \eqref{eq:f4} and \eqref{eq:f6} can be used to derive three further noncontextuality constraints:
\begin{align}
    \OnticXWeakMeasureAProb{j}(x|\lambda) &= (1 - \ProbabilityOfMeasurement)\frac{1}{2} + \ProbabilityOfMeasurement \OnticMeasureAProb{j}(x|\lambda),\\
    \sum_{x\in\{\pm1\}}  \OnticXWeakMeasureAProb{j}(x,\lambda'|\lambda) &= (1-\ProbabilityOfDisturbance) \delta(\lambda -\lambda') + \ProbabilityOfDisturbance \OnticDTransformation{j}(\lambda'|\lambda),\\
    \sum_{y\in\{\pm1\}}  \OnticYWeakMeasureAProb{j}(y,\lambda'|\lambda) &= (1-\ProbabilityOfDisturbance) \delta(\lambda -\lambda') + \ProbabilityOfDisturbance \OnticDTransformation{j}(\lambda'|\lambda) . \label{eq:NonContextualityConstraint4}
\end{align}
\noindent Here, $\OnticXWeakMeasureAProb{j}(x,\lambda'|\lambda)$ and $\OnticYWeakMeasureAProb{j}(y,\lambda'|\lambda)$ are the probability distributions representing the $X$-type and $Y$-type weak measurements of $\projA{j}$, respectively. Furthermore,  $\OnticXWeakMeasureAProb{j}(x|\lambda)$ is the probability distribution representing the $X$-type weak measurement of $\projA{j}$ followed by a discarding of the system; $\OnticMeasureAProb{j}(x|\lambda)$ is the probability distribution representing the measurement $\projA{j}$; and $\OnticDTransformation{j}(\lambda'|\lambda)$ is the transition matrix representing the quantum channel $\DChannel{j}$. Note II in the Supplementary Material provides detailed derivations of these constraints.

Any noncontextual hidden-variable model for the six protocols must satisfy Eqs. \eqref{eq:NonContextualityConstraint1} to \eqref{eq:NonContextualityConstraint4}. For certain states $\rho$, this might not be possible, making the realization of the protocols contextual. In the aforementioned protocols, the KD-nonpositivity $\NPr$ is a faithful witness of contextuality:
\begin{Theorem}\label{Thm:KDNegativityFaithfullyWitnessContextuality}  Assume that $\epsilon \ll 1$ in the six protocols.
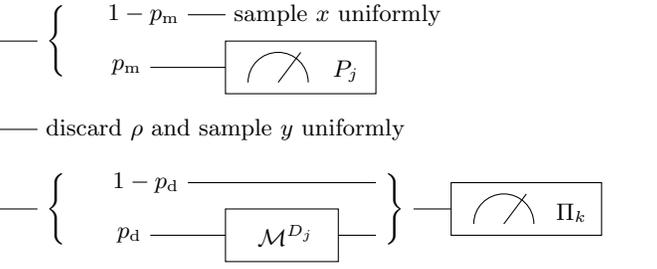
\begin{itemize}
    \item If $\rho$ is such that $\NPr>3d^2 \epsilon$, then these protocols do not admit a noncontextual hidden-variable model.
    \item If $\rho$ is KD-positive $[\NPr=0]$, then these protocols admit a noncontextual hidden-variable model.
\end{itemize}
\end{Theorem}

The theorem's first part is a KD-rephrasing of the result proven in \cite{Pusey14,kunjwal2019anomalous}. The original result is written in terms of weak values instead of the KD distribution. In Note III of the Supplementary Material, we provide a proof along similar lines, but in terms of the KD distribution, and with the exact bound on $\epsilon$. 

We prove the second part of the theorem by explicitly constructing a noncontextual hidden-variable model that has the same predictions as quantum theory [Eqs. \eqref{eq:f1} to \eqref{eq:f3}], and also obeys the Noncontextuality Constraints [Eqs. \eqref{eq:NonContextualityConstraint1} to \eqref{eq:NonContextualityConstraint4}]. The set of hidden variables we use is $\Lambda = \{ 1, 2, \ldots, d \}$, and we represent the preparation of $\rho$ by the outcome probability distribution when measuring the observable $B$: $\mu_{\rho}(\lambda = j) = \Tr(\rho \projB{j})$. In our noncontextual hidden-variable model, the other procedures are represented by the following probability distributions:
\begin{align}\label{eq:OnticMeasureBProbInHVM}
    \OnticMeasureBProb{k}(z|\lambda) = &
    \begin{cases}
        \delta_{k\lambda} &\text{ if } z = +1\\
        1 - \delta_{k\lambda} &\text{ if } z = -1,
    \end{cases} \\
    \label{eq:AnstazForWX}
    \OnticXWeakMeasureAProb{j}(x,\lambda'|\lambda) =& \delta_{\lambda\lambda'}\frac{1}{2}(1-\ProbabilityOfDisturbance)  + \frac{1}{2}\ProbabilityOfDisturbance \UndetermindedStochasticMatrix{j}(\lambda'|\lambda)\nonumber\\ &+ \frac{1}{2}x \ProbabilityOfMeasurement \delta_{\lambda\lambda'} (2\Re \frac{\KD{j,\lambda'}(\rho)}{\Tr(\projB{\lambda'}\rho)}-1), \\
    \label{eq:AnsatzForWY}
    \OnticYWeakMeasureAProb{j}(y,\lambda'|\lambda) = &  \delta_{\lambda\lambda'}\frac{1}{2}(1-\ProbabilityOfDisturbance)  + \frac{1}{2}\ProbabilityOfDisturbance \UndetermindedStochasticMatrix{j}(\lambda'|\lambda) 
\end{align}
Here, $\UndetermindedStochasticMatrix{j}(\lambda'|\lambda)$ is an arbitrary stochastic matrix. Straightforward analysis shows that these probability distributions reproduce  Predictions \eqref{eq:f1} to \eqref{eq:f3}. Furthermore, by substituting Eqs. \eqref{eq:OnticMeasureBProbInHVM} to \eqref{eq:AnsatzForWY} into the Noncontextuality Constraints \eqref{eq:NonContextualityConstraint1} to \eqref{eq:NonContextualityConstraint4}, one can confirm that it is possible to construct probability distributions $\OnticMeasureAProb{j}(x|\lambda)$ and $\OnticDTransformation{j}$ such that all Noncontextuality Constraints are satisfied; see Note IV of the Supplementary Material. 

\textit{Experiment and analysis:---}We now show our main result: We construct an experiment, the contextuality of which can be verified using only noncontextual procedures. We consider two experimenters, Alice and Bob, each performing an experiment:
\begin{enumerate}
    \item Alice chooses a sequence of $N$ pure quantum states $(\psi_j)_{j=1}^N$ (which may contain duplicates) such that these quantum states form an exotic state when mixed:
        \begin{equation}
            \exot = \frac{1}{N}\sum_{j=1}^N \psi_j.
        \end{equation}
    \item Alice sends the sequence of quantum states to Bob $M \gg 1$ times, each time keeping the order of the sequence secret. 
    \item For each state Bob receives, he randomly performs one of the six protocols.
    \item Bob publicly announces which protocol he chose and what outcome he obtained for each of the $N\times M$ states he received.
\end{enumerate}
Alice's experiment involves preparing pure states, sending them to Bob, and recording the outcome that Bob announces. Bob's experiment involves taking input states prepared by Alice, and conducting one of his six protocols. Figure \ref{fig:AliceAndBobExperiment} illustrates these experiments. 

Bob lacks knowledge of the order in which Alice has sent her states and thus effectively receives systems in the exotic state $\exot$. From his measurements, Bob obtains the outcome-probability distributions for each measurement procedure [Eqs. \eqref{eq:f1} to \eqref{eq:f6}] and determines $Q(\exot)$. Since $\mcl{N}(\exot)=0$, Bob concludes, via Theorem \ref{Thm:KDNegativityFaithfullyWitnessContextuality}, that there exists a noncontextual hidden-variable model describing his experiment.

However, knowledge of $Q(\exot)$ allows for the informationally complete reconstruction of $\exot$ \cite{arvidssonshukur2024properties}. Thus, Bob knows that $\exot$ is an exotic state. In Note V of the Supplementary material, we show that Bob can find a $\delta > 0 $ such that any pure state decomposition of $\exot$ has at least one state $\psi_-$ such that $\mcl{N}(\psi_-) > \delta $. The outcomes of Bob's measurements are public, and Alice can analyze them for the trials where she prepared $\psi_-$. Bob then checks if $\delta> 3d^2 \epsilon$. If so, using Theorem $1$, Bob deduces that there is no noncontextual hidden-variable model for Alice's postselected data. That is, Bob has verified that Alice must have performed a contextual experiment. 

\textit{Discussion:---}To summarize, we have constructed a noncontextual experiment (Bob's) in which an experimenter can verify that another experiment (Alice's) is contextual. In a sense, it is possible to verify the existence of a quantum-classical boundary from its classical side.

The nature of the exotic state $\exot$ (a KD-positive state that cannot be written as a convex combination of pure KD-positive states) is crucial to our result. We can compare the above-described scenario to one where Alice instead sends Bob states that average to the maximally mixed state $I/d$. Again, Bob finds his experiment to be noncontextual. However, because the maximally mixed state can be written as a convex combination of KD-positive states (for example, the eigenstates of $A$), Bob cannot deduce that Alice would have verified contextuality, and thus cannot himself verify contextuality.

Furthermore, Alice may be abstracted away in our setup. Bob effectively receives a mixed state because of his limited knowledge. Nevertheless, he may infer that information exists that could ``unmix'' the states he has received. Thus, Bob can verify that a contextual experiment will have happened for any observer that has access to this information.

Our result has an analogy in entanglement theory \cite{Vedral2006, Horodecki2009}. Consider a bipartite state $\rho_{\tt C,D}$ spatially split between Charlie and Dave.  If $\rho_{ \tt C,D}$ is not entangled, then there exists a local hidden-variable model that describes the outcomes of any measurement conducted by Charlie and Dave \cite{Horodecki2009}. The converse, however, is not true: The two subsystems can be entangled and yet admit a local hidden-variable model \cite{Werner1989, Wiseman2007}. Similarly, if Bob's experiment does not allow for the verification of quantum theory's contextuality, then it admits a noncontextual hidden-variable model. As we have shown, the converse, however, is not true: Bob's experiment can be used to verify contextuality and yet admit a noncontextual hidden-variable model.

We can also compare the contextuality witness of our analysis with that of entanglement theory. We witness contextuality with KD-nonpositivity $\NPr$. In entanglement theory, the von Neumann entropy $S(\rho_{\tt C}) = S(\rho_{\tt D})$, where $\rho_{\tt C} = \Tr_{ \tt D} (\rho_{\tt C,D})$ etc., quantifies the entanglement of a pure state $\rho_{\tt C,D}$ \cite{Neumann1927, Vedral2006, Horodecki2009}. The von Neumann entropy $S(\rho_{ \tt C})$ is a \textit{concave} function in $\rho_{\tt C,D}$. For some states, $S(\rho_{\tt C})$ exceeds $0$ even if $\rho_{\tt C,D}$ can be written as a convex combination of pure nonentangled states. Therefore, $S(\rho_{\tt C})$ is a poor measure of mixed-state entanglement. The total KD-nonpositivity $\NPr$ is a \textit{convex} function in $\rho$. For some (exotic) states $\exot$, $\mcl{N}(\exot)$  equals $0$ even if $\exot$ cannot be written as a convex combination of pure KD-positive states. As discussed above, this means that $\NPr$ can be a poor witness of contextuality in our settings. 

In entanglement theory, one can quantify the smallest pure-state-level entanglement of a mixed state $\rho_{\tt C,D}$ via the convex roof of the von Neumann entropy $S(\rho_{ \tt C})$ \cite{Hill1997, Wooters1998, Horodecki2009}. Similarly, one can quantify the smallest pure-state-level of KD-nonpositivity of a mixed state $\rho$ via the convex roof of $\NPr$ \cite{langrenez2024convexroofs}. As we have seen above, such a pure-state analysis can be useful for verifying contextuality. 

Finally,  our result assumes quantum theory. Bob needs to reconstruct the state from his data in order to conclude that it is exotic. To do this, Bob needs to derive Eqs. \eqref{eq:f1} to \eqref{eq:f3}, which requires quantum theory. Future research may strengthen our results to apply to general theories, beyond quantum theory.

\begin{acknowledgments}
Acknowledgments: JJT was supported by the Cambridge Trust. WS was supported by the EPRSC and Hitachi. This work was supported in part by the Agence Nationale de la Recherche under grant ANR-11-LABX-0007-01 (Labex CEMPI), by the Nord-Pas de Calais Regional Council and the European Regional Development Fund through the Contrat de Projets \'Etat-R\'egion (CPER), and by the CNRS through the MITI interdisciplinary programs. SDB thanks Girton College, where part of this work was performed, for its hospitality. The authors thank Christopher K. Long for his helpful comments.
\end{acknowledgments}

\bibliography{Bibliography}

\pagebreak
\widetext
\newpage
\begin{center}
\textbf{\large Supplementary Material for \\ Contextuality Can be Verified with Noncontextual Experiments}
\end{center}
\setcounter{equation}{0}
\setcounter{figure}{0}
\setcounter{table}{0}
\setcounter{page}{1}
\makeatletter
\renewcommand{\theequation}{S\arabic{equation}}
\setcounter{section}{0}
\renewcommand{\thesection}{\Roman{section}}

\section{The Quantum Theoretical Descriptions and Predictions of the Protocols}\label{app:Quantum_Theoretical_Predictions}

In this note, we describe the protocols in Figure 2 of the article in detail. First, we give an exact description of our weak measurement. We then calculate the quantum theoretical predictions of the outcomes of Protocols 1 to 3. Last, we show that the outcome distributions of Protocols 4 to 6 match those obtained by marginalizing the outcome distributions of Protocols 2 and 3. 

\begin{figure}[h]
    \centering
    \begin{quantikz}
        \lstick{$\ket{\psi_\epsilon}$} & &\gate[2]{U = I\otimes\projA{j} + Z\otimes\projA{j}^{\perp} } &&\meter{\ket{\pm}}\\
        \lstick{$\rho$}&& & &
    \end{quantikz}

    \caption{\textbf{Weak-Measurement Circuit.} This circuit implements the X-type weak measurement on a system in an arbitrary state $\rho$ for the observable $P_j$. The pointer is initialized in the state $\ket{\psi_\epsilon}\coloneqq \cos \epsilon \ket{0} + \sin \epsilon \ket{1}$, where $\epsilon$ controls the strength of the weak measurement. The pointer is entangled with the system via the unitary $U = I\otimes\projA{j} + Z\otimes\projA{j}^{\perp}$. The pointer is then measured in the $X$ basis, giving the outcome $x \in \{\pm 1\}$. The circuit for the Y-type weak measurement is identical, except that the final measurement on the ancilla is replaced by a measurement in the $Y$ basis instead of the $X$ basis. The Y-type weak measurement returns the outcome $y \in \{\pm 1\}$.}
    \label{fig:weak-value-measurement-circuit}
\end{figure}
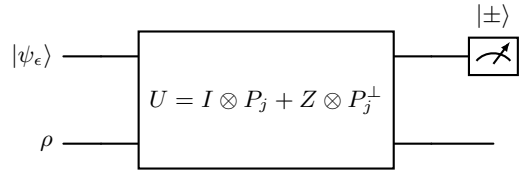

\subsection{The Weak Measurement}
We consider a weak measurement involving a qubit pointer system \cite{Jozsa2007}. The weak measurement works by entangling the system with an ancillary pointer system and then measuring the pointer in the $X$ or the $Y$ basis. The weak measurement thus returns a bit of information $x\in\{\pm 1\}$ or $y \in\{\pm1\}$. The strength of the weak measurement is given by the parameter $\epsilon$, which we take to be small. The details are given in Figure \ref{fig:weak-value-measurement-circuit}. 

In quantum theory, the weak measurement is represented by Kraus operators. For the X-type weak measurement given in Figure \ref{fig:weak-value-measurement-circuit}, the Kraus operators are
\begin{equation}\label{eq:KrausOperatorsXPointersup}
    \wvrkraus{x,j} \coloneqq \bra{X = x}U\ket{\psi_\epsilon} = \frac{1}{\sqrt{2}}(\cos \epsilon I + x \sin \epsilon \Doperator{j}),
\end{equation}
\noindent where $\Doperator{j} \coloneqq \projA{j} - \projA{j}^{\perp}$. Similarly, one may also calculate the Kraus operators for the Y-type weak measurement:
\begin{equation}\label{eq:KrausOperatorYPointersup}
        \wvikraus{y,j} \coloneqq \bra{Y = y}U\ket{\psi_\epsilon} = \frac{1}{\sqrt{2}}(\cos \epsilon I -i y \sin \epsilon \Doperator{j}).
\end{equation}

\subsection{Protocols 1, 2 and 3}
We now use quantum theory to calculate the outcome distributions for Protocols 1 to 3 in Figure 2 of the article. We denote the outcome probability distributions by $\OutcomePMF{1}{k}(z), \OutcomePMF{2}{j,k}(x,z)$ and $\OutcomePMF{3}{j,k}(y,z)$ for each protocol, respectively.

We start by calculating the outcome probabilities for Protocol 1. This protocol involves measuring whether the system is in the $k$th eigenspace of observable $B$, and then returning $+1$ if this is the case and $-1$ if this is not the case. Finding the outcome $+1$ thus corresponds to applying the projector $\projB{k}$, and finding the outcome $-1$ thus corresponds to applying the projector $\projB{k}^{\perp} = I - \projB{k}$. Recall that in the article, we introduced the notation
\begin{equation}
    \projB{k}^z = \begin{cases}
        \projB{k},      &z=+1,\\
        I - \projB{k},  &z=-1.
    \end{cases}
\end{equation}
The outcome probability distribution for Protocol 1 is thus
\begin{equation}
    \OutcomePMF{1}{k}(z) = \Tr(\projB{k}^{z} \rho) \coloneqq  \ProbOfBOnExot{k}{z}.
\end{equation}

We now calculate the outcome probability distribution $\OutcomePMF{2}{j,k}(x,z)$ for Protocol 2. The protocol involves performing an X-type weak measurement and then measuring whether or not the system is in the $k$th eigenspace of observable $B$, returning $+1$ if this is the case and $-1$ if it is not. The first measurement corresponds to applying the Kraus operators [see Eq. \eqref{eq:KrausOperatorsXPointersup}], and the second measurement corresponds to applying the same projectors as for Protocol 1. We find that 
\begin{align}
    \OutcomePMF{2}{j,k}(x,z)    &= \bP(x)\bP(z|x),\nonumber\\
                &= \Tr(\wvrkraus{x,j}\rho\wvrkraus{x,j}^\dag)\Tr\left[\projB{k}^z\frac{\wvrkraus{x,j}\rho\wvrkraus{x,j}^\dag}{\Tr(\wvrkraus{x,j}\rho\wvrkraus{x,j}^\dag)}\right],\nonumber\\
                &= \Tr(\projB{k}^z\wvrkraus{x,j}\rho\wvrkraus{x,j}^\dag),\nonumber\\
                &=\frac{1}{2}(1-\ProbabilityOfDisturbance)\ProbOfBOnExot{k}{z} + \frac{1}{2}x\ProbabilityOfMeasurement\ProbOfBOnExot{k}{z}[2\Re (w^z_{j,k}) - 1] + \frac{1}{2}\ProbabilityOfDisturbance \ProbOfBOnDExotD{j,k}{z}.\label{eq:f2sup}
\end{align}
\noindent Here, we have defined the probabilities $\ProbabilityOfMeasurement \coloneqq \sin(2\epsilon)$ and $\ProbabilityOfDisturbance \coloneqq \sin^2(\epsilon)$. Moreover, $\ProbOfBOnExot{k}{z}$ is the outcome probability of measuring whether or not the system is in the $k$th eigenspace, and $\ProbOfBOnDExotD{j,k}{z} \coloneqq \Tr(\projB{k}^{z} D_j \rho D_j^{\dagger})$ is the outcome probability of measuring whether or not the system is in the $k$th eigenspace of the observable $B$ after applying the $D_j$ channel to the state. We have also introduced the notation 
\begin{equation}
    w^z_{j,k} = 
    \begin{cases}
        w_{j,k} &\text{ if } z = 1 \\
        \frac{\sum_{m\neq k}w_{j,m}p_{+1}^{m}}{\sum_{m\neq k}p_{+1}^{m}} &\text{ if } z = 0,
    \end{cases}
\end{equation}
\noindent where $w_{j,k}$ is the weak-value matrix \cite{Aharanov88,Duck89,Jozsa2007,Dressel14} given by
\begin{equation}\label{eq:TheWeakValues}
    w_{j,k} \coloneqq \frac{\Tr( \projB{k}\projA{j}\rho)}{\Tr(\projB{k}\rho)} = \frac{\KD{j,k}(\rho)}{\Tr(\projB{k}\rho)}.
\end{equation}

The outcome probabilities for Protocol 3 can then be calculated analogously to those for Protocol 2, except that we need to use the Kraus operators $\wvikraus{y,j}$ defined in Eq. \eqref{eq:KrausOperatorYPointersup} instead of $\wvrkraus{x,j}$:
\begin{equation}\label{eq:f3sup}
    \OutcomePMF{3}{j,k}(y,z) = \Tr (\projB{k}^{z}\wvikraus{y,j}\rho\wvikraus{y,j}^{\dagger}) = \frac{1}{2}(1-\ProbabilityOfDisturbance)\ProbOfBOnExot{k}{z} + y\ProbabilityOfMeasurement\ProbOfBOnExot{k}{z}\Im(w^z_{j,k})  + \frac{1}{2}\ProbabilityOfDisturbance \ProbOfBOnDExotD{j,k}{z}.
\end{equation}

Equations (7) and (8) of the article then follow upon approximating the right-hand side in Eqs. \eqref{eq:f2sup} and \eqref{eq:f3sup} to first order in $\epsilon$ and using Eq. \eqref{eq:TheWeakValues} to write the result in terms of KD distributions.

\subsection{Protocols 4, 5 and 6}
We now show that, according to quantum theory, the outcome probability distributions for Protocols 4 to 6 can be obtained by marginalizing the outcome distributions of Protocols 2 and 3. Thus, if quantum theory is correct, this means that procedures (preparations, transformations and measurements) within these protocols are experimentally indistinguishable. Consequently, we can establish noncontextuality constraints for any noncontextual hidden-variable model of the protocols.

First, we show that summation of the $z$ variable in $\OutcomePMF{2}{j,k}(x,z)$ yields $\OutcomePMF{4}{j}(x)$. Tracing over the $z$ outcome in Procedure 2 corresponds to discarding the system after performing the X-type weak measurement. We thus effectively implement the following POVM elements by tracing over $z$:
\begin{equation}\label{eq:POVMOperatorsXWeakMeasurement}
    \wvrkraus{x,j}^{\dagger}\wvrkraus{x,j} = (1 - \ProbabilityOfMeasurement)\frac{I}{2} + \ProbabilityOfMeasurement \projA{j}^{x}.
\end{equation}
\noindent Here we introduced the notation $\projA{j}^{x} = {\projA{j}}$ if $x = +1$ and $\projA{j}^{x}$ $ = {\projA{j}^{\perp}}$ if $x = -1$. Thus, we find that measuring the weak value and discarding the system is equivalent to with probability $1 - \ProbabilityOfMeasurement$ returning $\pm1$ uniformly at random, and with probability $\ProbabilityOfMeasurement$ measuring whether or not the system is in the $j$th eigenspace of the operator $A$ and returning the result. Therefore, the right-hand side of Eq. \eqref{eq:POVMOperatorsXWeakMeasurement} exactly corresponds to the measurement process in Protocol 4:
\begin{equation}\label{eq:f2summedoverzgivesf4}
    \sum_{z}\OutcomePMF{2}{j,k}(x,z) = \OutcomePMF{4}{j}(x).
\end{equation}

We now proceed similarly to show that the summing over $z$ in $\OutcomePMF{3}{j,k}(y,z)$ yields $\OutcomePMF{5}{ }(y)$. As before, summing over the $z$ variable in Protocol 3 corresponds to the measurement process of discarding the system after performing a Y-type weak measurement. The POVM elements of this process are
\begin{equation}
    \wvikraus{y}\wvikraus{y}^{\dagger} = \frac{1}{2}I.
\end{equation}

The right-hand side of this equation corresponds to sampling $y$ uniformly at random. This is exactly what is done in Protocol $5$. We thus find that
\begin{equation}
    \sum_z \OutcomePMF{3}{j,k}(y,z) = \OutcomePMF{5}{ }(y).
\end{equation}

Next, we show that summing over $x$ in the outcome probability distribution $\OutcomePMF{2}{j,k}(x,z)$ yields the outcome probability distribution $\OutcomePMF{6}{j,k}(z)$. Consider Protocol $2$ without the projective measurement $\projB{k}$. Tracing over $x$ corresponds to implementing the quantum channel
\begin{equation}\label{eq:TheQuantumChannelInQT}
    \rho \mapsto \sum_x \wvrkraus{x,j}\rho\wvrkraus{x,j}^{\dagger} = (1-\ProbabilityOfDisturbance)\rho + \ProbabilityOfDisturbance \Doperator{j}\rho\Doperator{j}^{\dagger},
\end{equation}
\noindent  $D_j$ is unitary, so this is a well-defined quantum channel. The right-hand side of Eq. \eqref{eq:TheQuantumChannelInQT} is exactly the transformation process we apply in Protocol 6 before measuring $\projB{k}$. We thus find that
\begin{equation}
    \sum_x \OutcomePMF{2}{j,k}(x,z) = \OutcomePMF{6}{j,k }(z).
\end{equation}

A similar calculation can be used to show that
\begin{equation}
    \sum_y \OutcomePMF{3}{j,k}(y,z) = \sum_x \OutcomePMF{2}{j,k}(x,z) = \OutcomePMF{6}{j,k }(z).
\end{equation}

\section{Noncontextuality Constraints Established by Protocols 1 to 6}

\begin{table}
    \centering
        \begin{tabular}{|c|c|c|c|}
          \hline
          \textbf{Description}    &   \textbf{Pictogram}    &   \textbf{Quantum Model} & \textbf{Hidden-Variable Model}
          \\ \hline
          Prepare the  state &\makecell[c]{\vspace{5pt} \\ \hspace{5pt} 
          \begin{tikzpicture}[]
            \draw (0,0) node[left] {$\rho$} -- (1,0) ;
          \end{tikzpicture} \hspace{5pt} \\ \vspace{5pt} }  & Density Matrix $\rho$ & $\mu_{\rho}(\lambda)$
          \\ \hline
          \makecell{Measurement of whether \\the system is in the $j$th \\ eigenspace of the observable $A$} &
          \makecell[c]{\vspace{5pt} \\ \hspace{5pt} \begin{tikzpicture}[]
                \draw (-0.5,0.5) -- (0,0.5);
                \draw (0,0) rectangle (2,1);
                \draw (1.1,0.3) arc[start angle=0, end angle=180, radius=0.4];
                \draw (0.7,0.3) -- (1.0,0.7);
                \node at (1.6,0.45) {$\projA{j}$};
          \end{tikzpicture}\hspace{5pt} \\ \vspace{5pt}}      & \makecell{Projectors \\$\{\projA{j}, \projA{j}^{\perp} = I - \projA{j}\}$}&$\OnticMeasureAProb{j}(z|\lambda)$
          \\ \hline
          \makecell{Measurement of whether \\the system is in the $k$th \\ eigenspace of the observable $B$} &
          \makecell[c]{\vspace{5pt} \\ \hspace{5pt} \begin{tikzpicture}[]
                \draw (-0.5,0.5) -- (0,0.5);
                \draw (0,0) rectangle (2,1);
                \draw (1.1,0.3) arc[start angle=0, end angle=180, radius=0.4];
                \draw (0.7,0.3) -- (1.0,0.7);
                \node at (1.6,0.45) {$\projB{k}$};
          \end{tikzpicture}\hspace{5pt} \\ \vspace{5pt}}      & \makecell{Projectors \\$\{\projB{k}, \projB{k}^{\perp} = I - \projB{k}\}$}&$\OnticMeasureBProb{k}(z|\lambda)$
          \\ \hline
          \makecell{Weak measurement \\ of $\projA{j}$  using an \\ X-type pointer} &
          \makecell[c]{\vspace{5pt} \\ \hspace{5pt} \begin{tikzpicture}[]
                \draw (-0.5,0.5) -- (0,0.5);
                \draw (0,0) rectangle (2.3,1);
                \draw (1.1,0.3) arc[start angle=0, end angle=180, radius=0.4];
                \draw (0.7,0.3) -- (1.0,0.7);
                \node at (1.7,0.5) {$W^{(X)}_{\projA{j}}$};
                \draw (2.3,0.5) -- (2.8,0.5);
          \end{tikzpicture}\hspace{5pt} \\ \vspace{5pt}}      & \makecell{Kraus Operators \\$\{\wvrkraus{x,j}\}_{x\in \{\pm1\}}$}&$\OnticXWeakMeasureAProb{j}(x,\lambda'|\lambda)$
          \\ \hline
          \makecell{Weak measurement \\ of $\projA{j}$  using a \\ Y-type pointer} &
          \makecell[c]{\vspace{5pt} \\ \hspace{5pt} \begin{tikzpicture}[]
                \draw (-0.5,0.5) -- (0,0.5);
                \draw (0,0) rectangle (2.3,1);
                \draw (1.1,0.3) arc[start angle=0, end angle=180, radius=0.4];
                \draw (0.7,0.3) -- (1.0,0.7);
                \node at (1.7,0.5) {$W^{(Y)}_{\projA{j}}$};
                \draw (2.3,0.5) -- (2.8,0.5);
          \end{tikzpicture}\hspace{5pt} \\ \vspace{5pt}}      & \makecell{Kraus Operators \\$\{\wvikraus{y,k}\}_{y\in \{\pm1\}}$}&$\OnticYWeakMeasureAProb{j}(y,\lambda'|\lambda)$
          \\ \hline
          \makecell{Quantum Channel $\DChannel{j}$} &
          \makecell[c]{\vspace{5pt} \\ \hspace{5pt} \begin{tikzpicture}[]
                \draw (-0.5,0.5) -- (0,0.5);
                \draw (0,0) rectangle (1.5,1);
                \node at (0.8,0.5) {$\DChannel{j}$};
                \draw (1.5,0.5) -- (2,0.5);
          \end{tikzpicture}\hspace{5pt} \\ \vspace{5pt}}      & \makecell{Quantum Channel \\$\DChannel{j}(\rho) = D_j \rho D_j^{\dagger}$\\ where $D_j = \projA{j} - \projA{j}^{\perp}$} &$\OnticDTransformation{j}(\lambda'|\lambda)$
          \\ \hline
          \end{tabular}
    \caption{\textbf{List of all procedures.} Each row contains a description of the procedure, a pictogram representing the procedure, the corresponding mathematical object in quantum theory, and the corresponding probability distribution in the hidden-variable model.}
    \label{fig:Procedures}
\end{table}
Here, we derive the constraints [Eqs. (12) to (15) of the article] any noncontextual hidden-variable model of Protocols 1 to 6 must satisfy. Equation (13) of the article was already derived in the article. Here we will derive Eqs. (12), (14) and (15).

For clarity, in Table \ref{fig:Procedures}, we give a list of all the procedures that are used in our six protocols. Any hidden-variable model describing those protocols must associate probability distributions to all procedures given in this table. The notation we use for every probability distribution is also outlined in this table. With a slight abuse of notation, we write $\OnticXWeakMeasureAProb{j}(x|\lambda)$ for the probability distribution corresponding to discarding the system after an $X$-type weak measurement. Therefore,

\begin{equation}
    \OnticXWeakMeasureAProb{j}(x|\lambda) = \int_\Lambda d\lambda'\OnticXWeakMeasureAProb{j}(x,\lambda'|\lambda).
\end{equation}

We start by deriving Noncontextuality Constraint (12) in the article, which follows from Eq. (9) in the article:
\begin{equation}
    \OutcomePMF{4}{j}(x) = \sum_z \OutcomePMF{2}{j,k}(x,z).
\end{equation}

\noindent This equation expresses the fact that the Protocols do not distinguish between (i) the procedure of measuring $\projA{j}$ returning $x$ with probability $\ProbabilityOfMeasurement$ and sampling $x$ uniformly at random with probability $1-\ProbabilityOfMeasurement$, and (ii) performing the $X$-type weak measurement and discarding the system afterward. Any noncontextual hidden-variable model must therefore assign the same probability distribution to these procedures:
\begin{equation}
    \OnticXWeakMeasureAProb{j}(x|\lambda) = (1 - \ProbabilityOfMeasurement)\frac{1}{2} + \ProbabilityOfMeasurement \OnticMeasureAProb{j}(x|\lambda).
\end{equation}

Next, we derive the Noncontextuality Constraint (14) in the article, which follows from the first equality in Eq. (11) in the article:
\begin{equation}
    \OutcomePMF{6}{j,k}(z) = \sum_x \OutcomePMF{2}{j,k}(x,z).
\end{equation}
\noindent This equation expresses that the Protocols do not distinguish between (i) the transformation procedure of performing the $X$-type weak measurement and forgetting the result, and (ii) the transformation procedure corresponding to applying the quantum channel $\DChannel{j}(\rho)$ with probability $\ProbabilityOfDisturbance$ and applying the identity channel with probability $1-\ProbabilityOfDisturbance$. Any noncontextual hidden-variable model must therefore assign the same probability distribution to these procedures:
\begin{equation}
    \sum_{x\in\{\pm1\}}  \OnticXWeakMeasureAProb{j}(x,\lambda'|\lambda) = (1-\ProbabilityOfDisturbance) \delta(\lambda -\lambda') + \ProbabilityOfDisturbance \OnticDTransformation{j}(\lambda'|\lambda).
\end{equation}

Last, we derive the Noncontextuality Constraint (15) in the article, which follows from the second equality in Eq. (11) in the article:
\begin{equation}
    \OutcomePMF{6}{j,k}(z) = \sum_y \OutcomePMF{3}{j,k}(y,z).
\end{equation}
\noindent This equation implies that the Protocols do not distinguish between (i) the transformation procedure of performing the $Y$-type weak measurement and forgetting the result, and (ii) the transformation procedure corresponding to applying the quantum channel $\DChannel{j}(\rho)$ with probability $\ProbabilityOfDisturbance$ and applying the identity channel with probability $1-\ProbabilityOfDisturbance$. Any noncontextual hidden-variable model must therefore assign the same probability distribution to these procedures:
\begin{equation}
    \sum_{y\in\{\pm1\}}  \OnticYWeakMeasureAProb{j}(y,\lambda'|\lambda) = (1-\ProbabilityOfDisturbance) \delta(\lambda -\lambda') + \ProbabilityOfDisturbance \OnticDTransformation{j}(\lambda'|\lambda) .
\end{equation}

\section{KD-Nonpositivity Implies Contextuality} \label{app:boundingforwardsdirection}
In this Note, we prove the first item of Theorem 1 in the article; we show that KD-nonpositivity implies contextuality for sufficiently small weak-measurement strengths $\epsilon$. This result is very closely related to the result in Refs. \cite{Pusey14, kunjwal2019anomalous}. These previous works phrased their claims in terms of weak values and not in terms of KD distributions. Furthermore, since the previous results follow from a Taylor expansion, they do not give an explicit bound on how small $\epsilon$ must be. We require such a bound to prove the main result of the article. 

We begin by connecting KD-nonpositivity to the existence of sufficiently negative or complex elements of the KD distribution.

\begin{Lemma}
    Let $\delta > 0$. If $\mcl{N}(\rho) > \delta$, then there exist indices $j,k$ such that either $\Re \KD{j,k}(\rho) < - \frac{\delta}{3d^2}$ or $|\Im \KD{j,k}(\rho)| > \frac{\delta}{3d^2}$.
\end{Lemma}

\begin{proof}
    We prove the contrapositive statement: if $\rho$ is such that for all $j,k$, $\Re \KD{j,k}(\rho) \geq -\frac{\delta}{3d^2}$ and $|\Im \KD{j,k}(\rho)| \leq \frac{\delta}{3d^2}$, then $\mcl{N}(\rho) \leq \delta$. Using the fact that the KD distribution is normalized, we may write the nonpositivity as
    \begin{equation}
        \mcl{N}(\rho) = -1 + \sum_{j,k}|\KD{j,k}(\rho)| = \sum_{j,k}(|\KD{j,k}(\rho)| - \Re \KD{j,k}(\rho)).
    \end{equation}

    Applying the triangle inequality on the absolute value in the sum then yields
    \begin{equation}
        \mcl{N}(\rho) \leq \sum_{j,k}(|\Re\KD{j,k}(\rho)| - \Re \KD{j,k}(\rho) + |\Im \KD{j,k}(\rho)|).
    \end{equation}

    The result then follows upon using that $\Re \KD{j,k}(\rho) \geq -\frac{\delta}{3d^2}$ and $|\Im \KD{j,k}(\rho)| \leq \frac{\delta}{3d^2}$.
\end{proof}

Next, we show that if the KD distribution is complex, then the protocols must be contextual.

\begin{Lemma}\label{lem:eps_bound_imag_wv}
Suppose that for some state $\rho$, there exist $j,k$ so that $\Im \KD{j,k}(\rho) \neq 0$, then for $\epsilon < \min\{|\Im \KD{j,k}(\rho)|, \pi/4\}$, Protocols 1 to 6 do not admit a noncontextual hidden-variable model.
\end{Lemma}

\begin{proof}
    We first consider the case that $\Im \KD{j,k}(\rho) > 0$. Suppose there is a noncontextual hidden-variable model for the six protocols in Figure 2. We will show that for sufficiently small $\epsilon$ such a model cannot be consistent with the predictions of quantum theory. We start by upper bounding the prediction of the hidden-variable model for $f_3(+1,+1)$. For $\lambda \in \Lambda$, let
    \begin{equation}
        \Ilamb = \{\lambda' \in \Lambda | \OnticMeasureBProb{k}(+1|\lambda') \leq \OnticMeasureBProb{k}(+1|\lambda)\},
    \end{equation}
    \noindent be the set of ontic states with a higher ``outcome $+1$" probability than $\lambda$. Using the hidden-variable model, we find that
    \begin{align}
        f_3(+1,+1) &= \int_\Lambda d\lambda \int_\Lambda d\lambda' \oprep[\rho](\lambda) \OnticYWeakMeasureAProb{j}(+1,\lambda'|\lambda)\OnticMeasureBProb{k}(+1|\lambda')\\
        &= \int_\Lambda d\lambda  \oprep[\rho](\lambda) \left[ \int_{\Ilamb} d\lambda'\OnticYWeakMeasureAProb{j}(+1,\lambda'|\lambda)\OnticMeasureBProb{k}(+1|\lambda') +  \int_{\Lambda \setminus \Ilamb}  d\lambda'\OnticYWeakMeasureAProb{j}(+1,\lambda'|\lambda)\OnticMeasureBProb{k}(+1|\lambda')  \right]\\
        &\leq \int_\Lambda d\lambda  \oprep[\rho](\lambda) \left[ \int_{\Ilamb} d\lambda'\OnticYWeakMeasureAProb{j}(+1,\lambda'|\lambda)\OnticMeasureBProb{k}(+1|\lambda) +  \int_{\Lambda \setminus \Ilamb}  d\lambda'\OnticYWeakMeasureAProb{j}(+1,\lambda'|\lambda)  \right]\\
        &\leq \int_\Lambda d\lambda  \oprep[\rho](\lambda) \left[ \int_{\Lambda} d\lambda'\OnticYWeakMeasureAProb{j}(+1,\lambda'|\lambda)\OnticMeasureBProb{k}(+1|\lambda) +  \int_{\Lambda \setminus \Ilamb}  d\lambda' \sum_y\OnticYWeakMeasureAProb{j}(y,\lambda'|\lambda)  \right]. \label{eq:LastLineOfUpperBoundingOnticf3}
    \end{align}
    \noindent In the second line we broke up the integral; in the third line we used the definition of $\Ilamb$ for the first term and that $\OnticMeasureBProb{k}(+1|\lambda') \leq 1$ for the second term; and on the fourth line we increased the set integrated over for the first term and we summed over $y$ in the second term.

    We proceed by substituting Noncontextuality Constraints (12) and (15) from the article into Eq. \eqref{eq:LastLineOfUpperBoundingOnticf3}. As $\lambda \notin \Lambda\setminus\Ilamb$,  the delta function in Eq. (15) of the article vanishes inside the integral. We obtain
    \begin{align}
        f_3(+1,+1) &\leq \int_\Lambda d\lambda  \oprep[\rho](\lambda) \left[ \frac{1}{2}\OnticMeasureBProb{k}(+1|\lambda) +  \ProbabilityOfDisturbance\int_{\Lambda \setminus \Ilamb}  d\lambda' \OnticDTransformation{j}(\lambda'|\lambda)  \right]\\
        &\leq \int_\Lambda d\lambda  \oprep[\rho](\lambda) \left[ \frac{1}{2}\OnticMeasureBProb{k}(+1|\lambda) +  \ProbabilityOfDisturbance  \right]\\
        &= \frac{1}{2}p_{+1}^k + \ProbabilityOfDisturbance, \label{eq:UpperBoundForf3}
    \end{align}
    \noindent where in the second line we used that integrating over any set of outcomes of a transformation matrix yields a probability less than $1$. 

    We now turn to the quantum theoretical prediction for $f_3(+1,+1)$. Equation \eqref{eq:f3sup} dictates that it is
    \begin{align}
        f_3(+1,+1) &= \frac{1}{2}(1-\ProbabilityOfDisturbance)p_{+1}^{k} + \ProbabilityOfMeasurement p_{+1}^k \Im w_{j,k} + \frac{1}{2}\ProbOfBOnDExotD{j,k}{z}\\
        & \geq \frac{1}{2}(1-\ProbabilityOfDisturbance)p_{+1}^{k} + \ProbabilityOfMeasurement \Im Q_{j,k}(\rho). \label{eq:LowerBoundForf3}
    \end{align}

    Comparing Eqs. \eqref{eq:UpperBoundForf3} and $\eqref{eq:LowerBoundForf3}$, we find that 
    \begin{equation}
        \frac{1}{2}(1-\ProbabilityOfDisturbance)p_{+1}^{k} + \ProbabilityOfMeasurement \Im Q_{j,k}(\rho) \leq \frac{1}{2}p_{+1}^k + \ProbabilityOfDisturbance.
    \end{equation}

    Using that $\ProbabilityOfDisturbance = \sin^2\epsilon$ and that $\ProbabilityOfMeasurement = \sin 2\epsilon$, and rearranging we get that
    \begin{equation}
        \tan \epsilon \geq \frac{4}{2+p_{+1}^{k}}\Im \KD{j,k}(\rho) \geq \frac{4}{\pi} \Im Q_{j,k}(\rho),
    \end{equation}
    \noindent where in the last line we used that $2+p_{+1}^k \leq \pi $ and that $\Im \KD{j,k}(\rho) > 0$. For $\epsilon \in [0,\pi/4]$, $\tan \epsilon \leq \frac{4\epsilon}{\pi}$. Thus, for $\epsilon \in [0,\pi/4]$, and under the assumption of a noncontextual hidden-variable model, consistent with quantum theory, we find the inequality 
    \begin{equation}
        \epsilon \geq \Im \KD{j,k}(\rho).
    \end{equation}

    This inequality is in contradiction with the hypothesis that $\epsilon < \min\{\Im \KD{j,k}(\rho), \pi/4\}$. We conclude that a noncontextual hidden-variable model that is consistent with quantum theory and describes Protocols 1 to 6 cannot exist for such $\epsilon$.

    The proof for the case $\Im \KD{j,k}(\rho) < 0$ follows from considering $f_3(-1,+1)$ and following exactly the same steps.
    
\end{proof}

A similar result holds for the case where the KD distribution is negative:

\begin{Lemma}\label{lem:eps_bound_real_wv}
Suppose that for some state $\rho$, and some $j,k$, $\Re \KD{j,k}(\rho) < 0$, then for $\epsilon < \min\{-\Re \KD{j,k}(\rho), \pi/4\}$, Protocols 1 to 6 do not admit a noncontextual hidden-variable model.    
\end{Lemma}

\begin{proof}
    The proof of this lemma is similar to the one of Lemma \ref{lem:eps_bound_imag_wv}. We assume a noncontextual hidden-variable model exists and find a contradiction. We start by upper bounding $f_2(-1,+1)$ using the hidden-variable model. We define $\Ilamb$ as in the proof of Lemma \ref{lem:eps_bound_imag_wv}, and through the exact same steps find that
    \begin{equation}\label{eq:UpperBoundOnf2}
        f_2(-1,+1) \leq \int_\Lambda d\lambda  \oprep[\rho](\lambda) \left[ \int_{\Lambda} d\lambda'\OnticXWeakMeasureAProb{j}(-1,\lambda'|\lambda)\OnticMeasureBProb{k}(+1|\lambda) +  \int_{\Lambda \setminus \Ilamb}  d\lambda' \sum_x\OnticXWeakMeasureAProb{j}(x,\lambda'|\lambda)  \right].
    \end{equation}

    Next, we plug Noncontextuality Constraints (13) and (14) from the article into Eq. \eqref{eq:UpperBoundOnf2}. Using the same steps as we did for Lemma \ref{lem:eps_bound_imag_wv}, we find that
    \begin{equation}\label{eq:UpperboundforF2two}
        f_2(-1,+1) \leq \frac{1}{2}(1+\ProbabilityOfMeasurement)p_{+1}^k +\ProbabilityOfDisturbance
    \end{equation}

    Next, we consider the prediction by quantum theory, given by Eq. \eqref{eq:f2sup}:
    \begin{align}
        f_2(-1,+1) &= \frac{1}{2}(1-\ProbabilityOfDisturbance) p_{+1}^k - \frac{1}{2}\ProbabilityOfMeasurement p_{+1}^k (2\Re w_{j,k} -1) + \frac{1}{2}\ProbabilityOfDisturbance\ProbOfBOnDExotD{j,k}{+1}\\
        &\geq \frac{1}{2}(1-\ProbabilityOfDisturbance) p_{+1}^k - \ProbabilityOfMeasurement\Re \KD{j,k}(\rho)  + \frac{1}{2}\ProbabilityOfMeasurement p_{+1}^k \label{eq:LowerBoundForF2}
    \end{align}

    We now combine Eqs. \eqref{eq:UpperboundforF2two} and \eqref{eq:LowerBoundForF2} and use that $\ProbabilityOfDisturbance = \sin^2 \epsilon$ and that $\ProbabilityOfMeasurement = \sin 2\epsilon$ to get that
    \begin{equation}
        \tan \epsilon \geq -\frac{4}{2+p_{+1}^k} \Re \KD{j,k}(\rho) \geq -\frac{4}{\pi} \KD{j,k}(\rho).
    \end{equation}

    Again, for $\epsilon \in [0,\pi/4]$, $\tan \epsilon \leq \frac{\pi}{4}\epsilon$, so we get that the existence of a noncontextual hidden-variable model implies that 
    \begin{equation}
        \epsilon \geq -\Re \KD{j,k}(\rho),
    \end{equation}
    from which the Lemma follows.

\end{proof}

Combining these Lemmata, we find:
\begin{Theorem}
    Let $\delta \in (0,\pi/4]$. If $\mcl{N}(\rho) > \delta$, then for all $\epsilon < \frac{\delta}{3d^2}$, the Protocols 1 to 6 do not admit a noncontextual hidden-variable model.
\end{Theorem}

\section{KD-Positivity Implies Noncontextuality}\label{app:Non-Anomalous Weak Values Imply Non-Contextuality}

In this Note, we prove the second item in Theorem 1 of the article: for sufficiently small $\epsilon$, KD-positivity implies that there exists a noncontextual hidden variable model. Let us formally restate what we are proving here:

\begin{Lemma}
    If $\rho$ is KD-positive and $\epsilon < \frac{\sqrt{5}}{5}$, then there exists a noncontextual hidden variable model for Protocols 1 to 6.
\end{Lemma}

We will prove this by explicitly constructing a noncontextual hidden variable model that correctly reproduces the outcomes of Protocols 1 to 6 as predicted by quantum mechanics.  We will use the notation set out in Table~\ref{fig:Procedures}. Any such hidden-variable model must be consistent with the predictions of quantum theory (Eqs. (6) to (8) of the article). We thus get the following `Correctness Constraints':
\begin{align}
    \sum_{\lambda \in \Lambda} \oprep[\rho](\lambda) \OnticMeasureBProb{k}(z|\lambda) &= \OutcomePMF{1}{k}(z),\label{eq:CorrectnessConstraint1}\\
    \sum_{\lambda,\lambda' \in \Lambda} \oprep[\rho](\lambda) \OnticXWeakMeasureAProb{j}(x, \lambda'|\lambda) \OnticMeasureBProb{k}(z|\lambda') &= \OutcomePMF{2}{j,k}(x,z),\label{eq:CorrectnessConstraint2}\\
    \sum_{\lambda,\lambda' \in \Lambda} \oprep[\rho](\lambda)    \OnticYWeakMeasureAProb{j}(y, \lambda'|\lambda) \OnticMeasureBProb{k}(z|\lambda') &= \OutcomePMF{3}{j,k}(y,z),\label{eq:CorrectnessConstraint3}
\end{align}
\noindent where we have chosen the ontic space $\Lambda$ to be finite, turning the integrals into sums.

To ensure the model is noncontextual, we must ensure the hidden variable model satisfies the `Noncontextuality Constraints' described in the article (Eqs. (12) - (15) in the article). They are 
\begin{align}
    \sum_{\lambda'\in\Lambda} \OnticXWeakMeasureAProb{j}(x,\lambda'|\lambda) &= (1 - \ProbabilityOfMeasurement)\frac{1}{2} + \ProbabilityOfMeasurement \OnticMeasureAProb{j}(x|\lambda),\label{eq:NonContextualityConstraint1Sup}\\
    \sum_{x\in\{\pm1\}}  \OnticXWeakMeasureAProb{j}(x,\lambda'|\lambda) &= (1-\ProbabilityOfDisturbance) \delta_{\lambda\lambda'} + \ProbabilityOfDisturbance \OnticDTransformation{j}(\lambda'|\lambda),\label{eq:NonContextualityConstraint2}\\
    \sum_{\lambda'\in\Lambda} \OnticYWeakMeasureAProb{j}(y,\lambda'|\lambda) &= \frac{1}{2},\label{eq:NonContextualityConstraint3}\\
    \sum_{y\in\{\pm1\}}  \OnticYWeakMeasureAProb{j}(y,\lambda'|\lambda) &= (1-\ProbabilityOfDisturbance) \delta_{\lambda\lambda'} + \ProbabilityOfDisturbance \OnticDTransformation{j}(\lambda'|\lambda).\label{eq:NonContextualityConstraint4Sup}    
\end{align}

To construct a noncontextual hidden variable model, we must choose $\oprep[\rho],\OnticMeasureBProb{k}, \OnticMeasureAProb{j},\OnticXWeakMeasureAProb{j},\OnticYWeakMeasureAProb{j}$ and $\OnticDTransformation{j}$ such that Eqs. (\ref{eq:CorrectnessConstraint1}) - (\ref{eq:NonContextualityConstraint4Sup}) are satisfied for all $j,k$. Furthermore, we must ensure that the $\oprep[\rho],\OnticMeasureBProb{k}, \OnticMeasureAProb{j},\OnticXWeakMeasureAProb{j},\OnticYWeakMeasureAProb{j}$ and $\OnticDTransformation{j}$ we choose are probability distributions in the sense that they are normalized to $1$ and take values in $[0,1]$. We will prove that a noncontextual hidden variable model exists by explicitly providing $\oprep[\rho],\OnticMeasureBProb{k}, \OnticMeasureAProb{j},\OnticXWeakMeasureAProb{j},\OnticYWeakMeasureAProb{j}$ and $\OnticDTransformation{j}$ such that these constraints are satisfied.

Let us start by choosing our ontic space. We will set $\Lambda = \mathbb{Z}_d$. Our intuition will be that the hidden variable model simply stores what the outcome of the $B_k$ measurement will be. We thus force $\OnticMeasureBProb{k}$ to be outcome deterministic:
\begin{equation}\label{eq:OnticMeasureBProbInHVMSupp}
    \OnticMeasureBProb{k}(z|\lambda) = 
    \begin{cases}
        \delta_{k\lambda} &\text{ if } z = +1\\
        1 - \delta_{k\lambda} &\text{ if } z = {-1},
    \end{cases}
\end{equation}
\noindent which is clearly a normalized probability distribution. Next, we set $\oprep[\rho]$ to be the probability distribution from measuring $B$. Thus, using the notation introduced in Note \ref{app:Quantum_Theoretical_Predictions}, we set
\begin{equation}\label{eq:OprepInHVM}
    \oprep[\rho](\lambda) = p_{+1}^\lambda.
\end{equation}

We can now check Eq. (\ref{eq:CorrectnessConstraint1}) for $z = +1$ by plugging in Eqs. (\ref{eq:OnticMeasureBProbInHVMSupp}) and (\ref{eq:OprepInHVM}):
\begin{align}
\begin{split}
    \sum_{\lambda \in \Lambda} \oprep[\rho](\lambda) \OnticMeasureBProb{k}(+1|\lambda) &= \sum_{\lambda \in \mathbb{Z}_d} p_{+1}^\lambda \delta_{k\lambda}\\
    &= p_{+1}^k\\
    &= \OutcomePMF{1}{k}(+1),
\end{split}
\end{align}
\noindent which is correct. Furthermore, we may also check (\ref{eq:CorrectnessConstraint1}) for $z = {-1}$:
\begin{align}
\begin{split}
    \sum_{\lambda \in \Lambda} \oprep[\rho](\lambda) \OnticMeasureBProb{k}({-1}|\lambda) &= \sum_{\lambda \in \mathbb{Z}_d} p_{+1}^\lambda (1-\delta_{k\lambda})\\
    &= 1-p_{+1}^k\\
    &= p_{-1}^k\\
    &= \OutcomePMF{1}{k}({-1}),
\end{split}
\end{align}
\noindent which is also as required. Next, to check Eq. (\ref{eq:CorrectnessConstraint2}), we must provide $\OnticXWeakMeasureAProb{j}$. We use the following ansatz:
\begin{equation}\label{eq:AnstazForWXSupp}
    \OnticXWeakMeasureAProb{j}(x,\lambda'|\lambda) = \delta_{\lambda\lambda'}[\frac{1}{2}(1-\ProbabilityOfDisturbance) + \frac{1}{2}x \ProbabilityOfMeasurement (2\Re w_{j,\lambda'}-1)] + \frac{1}{2}\ProbabilityOfDisturbance \UndetermindedStochasticMatrix{j}(\lambda'|\lambda),
\end{equation}
\noindent where $\UndetermindedStochasticMatrix{j}$ is a stochastic matrix that we choose such that it is normalized (i.e. $\sum_{\lambda'} \UndetermindedStochasticMatrix{j}(\lambda'|\lambda) = 1$) and that it satisfies
\begin{equation}
    \sum_\lambda \UndetermindedStochasticMatrix{j}(\lambda'|\lambda) p_{+1}^{\lambda} = \ProbOfBOnDExotD{j\lambda'}{+1}.
\end{equation}

Such a $\UndetermindedStochasticMatrix{j}$ exists, as we may for example set $\UndetermindedStochasticMatrix{j}(\lambda'|\lambda) = \ProbOfBOnDExotD{j\lambda'}{+1}$. We must now ensure that this ansatz yields a probability distribution. We first check normalization:
\begin{equation}
    \begin{split}
        \sum_{\lambda' \in \mathbb{Z}_d, x \in {\pm 1}} \OnticXWeakMeasureAProb{j}(x,\lambda'|\lambda) &= \sum_{\lambda' \in \mathbb{Z}_d, x \in {\pm 1}}\left( \delta_{\lambda\lambda'}\left[\frac{1}{2}(1-\ProbabilityOfDisturbance) + \frac{1}{2}x \ProbabilityOfMeasurement (2\Re w_{j,\lambda'}-1)\right] + \frac{1}{2} \ProbabilityOfDisturbance \UndetermindedStochasticMatrix{j}(\lambda'|\lambda)\right)\\
        &= \sum_{\lambda' \in \mathbb{Z}_d} [\delta_{\lambda\lambda'}(1-\ProbabilityOfDisturbance)  + \ProbabilityOfDisturbance \UndetermindedStochasticMatrix{j}(\lambda'|\lambda)]\\
        &= 1 - \ProbabilityOfDisturbance + \ProbabilityOfDisturbance = 1.
    \end{split}
\end{equation}

Next, we show that if $\epsilon < \frac{\sqrt{5}}{5}$, then $\OnticXWeakMeasureAProb{j}(x,\lambda'|\lambda) \in [0,1]$.

\begin{Lemma}
    If $\epsilon < \frac{\sqrt{5}}{5}$, then $\OnticXWeakMeasureAProb{j}(x,\lambda'|\lambda)$ as given in Eq. (\ref{eq:AnstazForWXSupp}) lies in $[0,1]$.
\end{Lemma}

\begin{proof}
    $\OnticXWeakMeasureAProb{j}(x,\lambda'|\lambda)$ is clearly real, so it suffices to show that $\OnticXWeakMeasureAProb{j}(x,\lambda'|\lambda)\geq 0$ since normalization then implies that $\OnticXWeakMeasureAProb{j}(x,\lambda'|\lambda) \leq 1$. To show this, note that
    \begin{equation}\label{eq:boundingOnticMeasureAForXPointer}
        \begin{split}
            \OnticXWeakMeasureAProb{j}(x,\lambda'|\lambda) &= \delta_{\lambda\lambda'}[\frac{1}{2}(1-\ProbabilityOfDisturbance) + \frac{1}{2}x \ProbabilityOfMeasurement (2\Re w_{j,\lambda'}-1)] + \frac{1}{2}\ProbabilityOfDisturbance \UndetermindedStochasticMatrix{j}(\lambda'|\lambda)\\
            &\geq \delta_{\lambda\lambda'}[\frac{1}{2}(1-\ProbabilityOfDisturbance) + \frac{1}{2}x \ProbabilityOfMeasurement (2\Re w_{j,\lambda'}-1)]\\
            &\geq \delta_{\lambda\lambda'}[\frac{1}{2}\cos^2\epsilon - \frac{1}{2} \sin 2\epsilon ]\\
            &= \delta_{\lambda\lambda'}h(\epsilon),
        \end{split}
    \end{equation}
    \noindent where we defined $h(\epsilon) \coloneqq\frac{1}{2}\cos^2\epsilon - \frac{1}{2} \sin 2\epsilon$. By the Mean Value Theorem, there exists a $c \in [0,\epsilon]$ such that 
    \begin{equation}\label{eq:MVTApplied}
        h'(c)\epsilon + h(0) = h(\epsilon).
    \end{equation}

    There exists some phase $\psi$ such that
    \begin{equation}\label{eq:derivativeofh}
        h'(c) = \frac{1}{2}\sin 2c  - \cos 2c = \frac{\sqrt{5}}{2} \sin(2c + \psi) \geq -\frac{\sqrt{5}}{2}.
    \end{equation}

    Substituting \eqref{eq:MVTApplied} into (\ref{eq:boundingOnticMeasureAForXPointer}) and using Eq. \eqref{eq:derivativeofh} and that $h(0) = \frac{1}{2}$ yields
        \begin{equation}
        \begin{split}
            \OnticXWeakMeasureAProb{j}(x,\lambda'|\lambda) &= \delta_{\lambda\lambda'}[ h'(c)\epsilon + h(0)]\\
            &\geq \delta_{\lambda\lambda'}[\frac{-\sqrt{5}}{2} \epsilon + \frac{1}{2}]\\
            &\geq 0
        \end{split}
    \end{equation}
    \noindent where in the last line we used that $\epsilon < \frac{\sqrt{5}}{5}$. This completes the proof.
\end{proof}

We have thus found that the ansatz (\ref{eq:AnstazForWXSupp}) is a normalized probability distribution. We proceed to check that it indeed satisfies Eq. (\ref{eq:CorrectnessConstraint2}) for $z = +1$.

\begin{equation}
    \begin{split}
        \sum_{\lambda,\lambda' \in \Lambda} \oprep[\rho](\lambda) \OnticXWeakMeasureAProb{j}(x, \lambda'|\lambda) \OnticMeasureBProb{k}(+1|\lambda') &= \sum_{\lambda,\lambda' \in \mathbb{Z}_d} p_{+1}^{\lambda} \OnticXWeakMeasureAProb{j}(x, \lambda'|\lambda) \delta_{k\lambda'}\\
        &= \sum_{\lambda\in \mathbb{Z}_d} p_{+1}^{\lambda} \OnticXWeakMeasureAProb{j}(x, k|\lambda)\\
        &= \sum_{\lambda\in \mathbb{Z}_d} p_{+1}^{\lambda} \left\{\delta_{\lambda k}\left[\frac{1}{2}(1-\ProbabilityOfDisturbance) + \frac{1}{2}x \ProbabilityOfMeasurement (2\Re w_{j,k}-1)\right] + \frac{1}{2}\ProbabilityOfDisturbance \UndetermindedStochasticMatrix{j}(k|\lambda)\right\}\\
        &= \frac{1}{2}(1-\ProbabilityOfDisturbance)p_{+1}^{k}  + \frac{1}{2}x \ProbabilityOfMeasurement p_{+1}^{k} (2\Re w_{j,k}-1) + \frac{1}{2}\ProbabilityOfDisturbance \sum_{\lambda\in \mathbb{Z}_d} \UndetermindedStochasticMatrix{j}(k|\lambda)p_{+1}^{\lambda}\\
        &= \frac{1}{2}(1-\ProbabilityOfDisturbance)p_{+1}^{k}  + \frac{1}{2}x \ProbabilityOfMeasurement p_{+1}^{k} (2\Re w_{j,k}-1) + \frac{1}{2}\ProbabilityOfDisturbance \ProbOfBOnDExotD{j,k}{+1}\\
        &= \OutcomePMF{2}{j,k}(x,+1),
    \end{split}
\end{equation}
\noindent which is correct. Similarly, we may check Eq. (\ref{eq:CorrectnessConstraint2}) for $z = {-1}$:
\begin{equation}
    \begin{split}
        \sum_{\lambda,\lambda' \in \Lambda} \oprep[\rho](\lambda) \OnticXWeakMeasureAProb{j}(x, \lambda'|\lambda) \OnticMeasureBProb{k}({-1}|\lambda') &= \sum_{\lambda,\lambda' \in \mathbb{Z}_d} p_S^{\lambda} \OnticXWeakMeasureAProb{j}(x, \lambda'|\lambda) (1-\delta_{k\lambda'})\\
        &= \sum_{\lambda \in \mathbb{Z}_d, \lambda' \in \mathbb{Z}_d\backslash\{k\}} p_S^{\lambda} \OnticXWeakMeasureAProb{j}(x, \lambda'|\lambda) \\
        &= \sum_{\lambda' \in \mathbb{Z}_d\backslash\{k\}} \frac{1}{2}(1-\ProbabilityOfDisturbance)p_S^{\lambda'}  + \frac{1}{2}x \ProbabilityOfMeasurement p_S^{\lambda'} (2\Re w_{j\lambda'}-1) + \frac{1}{2}\ProbabilityOfDisturbance \ProbOfBOnDExotD{j\lambda'}{S}\\
        &= \frac{1}{2}(1-\ProbabilityOfDisturbance)p_{-1}^{k}  + \frac{1}{2}x \ProbabilityOfMeasurement p_{-1}^{k} \left(2\frac{\sum_{\lambda' \in \mathbb{Z}_d\backslash\{k\}} p_S^{\lambda'}\Re w_{j\lambda'}}{p_{-1}^{k}}-1\right) + \frac{1}{2}\ProbabilityOfDisturbance \ProbOfBOnDExotD{jk}{{-1}}\\
        &= \frac{1}{2}(1-\ProbabilityOfDisturbance)p_{-1}^{k}  + \frac{1}{2}x \ProbabilityOfMeasurement p_{-1}^{k} (2 g_{-1}(\Re w_{jk})-1) + \frac{1}{2}\ProbabilityOfDisturbance \ProbOfBOnDExotD{jk}{{-1}}\\
        &= \OutcomePMF{2}{j,k}(x,{-1}),
    \end{split}
\end{equation}
\noindent which is also correct. We thus conclude that our ansatz satisfies Eq. (\ref{eq:CorrectnessConstraint2}). 

We move on to checking the Noncontextuality Constraints involving $\OnticXWeakMeasureAProb{j}$. We start with Eq. (\ref{eq:NonContextualityConstraint1Sup}) for $x = +1$.
\begin{equation}
    \begin{split}
        \sum_{\lambda'\in\Lambda} \OnticXWeakMeasureAProb{j}(+1,\lambda'|\lambda) &= \sum_{\lambda'\in\mathbb{Z}_d} \left( \delta_{\lambda\lambda'}\left[\frac{1}{2}(1-\ProbabilityOfDisturbance) + \frac{1}{2} \ProbabilityOfMeasurement (2\Re w_{j,\lambda'}-1)\right] + \frac{1}{2}\ProbabilityOfDisturbance \UndetermindedStochasticMatrix{j}(\lambda'|\lambda)\right)\\
        &= \frac{1}{2}(1-\ProbabilityOfDisturbance) + \frac{1}{2} \ProbabilityOfMeasurement (2\Re w_{j,\lambda}-1) + \frac{1}{2}\ProbabilityOfDisturbance \\
        &= (1-\ProbabilityOfMeasurement)\frac{1}{2} + \ProbabilityOfMeasurement \Re w_{j,\lambda}
    \end{split}
\end{equation}

Comparing this with the right-hand side of Eq. (\ref{eq:NonContextualityConstraint1Sup}) for $x = +1$, we see that we must set $\OnticMeasureAProb{j}(+1|\lambda) = \Re w_{j,\lambda}$. This yields probabilities in $[0,1]$ since the fact that $\rho$ is KD-positive implies that the weak values lie in $[0,1]$ by Eq. \eqref{eq:TheWeakValues}. Furthermore, we fix the  $\OnticMeasureAProb{j}(+1|\lambda)$ via normalization: $\OnticMeasureAProb{j}(-1|\lambda) = 1 - \Re w_{j,\lambda}$. We can now also check Eq. (\ref{eq:NonContextualityConstraint1Sup}) for $x = -1$:
\begin{equation}
    \begin{split}
        \sum_{\lambda'\in\Lambda} \OnticXWeakMeasureAProb{j}(-1,\lambda'|\lambda) &= \sum_{\lambda'\in\mathbb{Z}_d} \left(\delta_{\lambda\lambda'}\left[\frac{1}{2}(1-\ProbabilityOfDisturbance) - \frac{1}{2} \ProbabilityOfMeasurement (2\Re w_{j,\lambda'}-1)\right] + \frac{1}{2}\ProbabilityOfDisturbance \UndetermindedStochasticMatrix{j}(\lambda'|\lambda)\right)\\
        &= \frac{1}{2}(1-\ProbabilityOfDisturbance) - \frac{1}{2} \ProbabilityOfMeasurement (2\Re w_{j,\lambda}-1) + \frac{1}{2}\ProbabilityOfDisturbance \\
        &= (1-\ProbabilityOfMeasurement) \frac{1}{2}- \ProbabilityOfMeasurement (1 - \Re w_{j,\lambda}) \\
        &= (1-\ProbabilityOfMeasurement) \frac{1}{2}- \ProbabilityOfMeasurement \OnticMeasureAProb{j}(-1|\lambda) 
    \end{split}
\end{equation}
\noindent as required. We thus conclude that Noncontextuality Constraint (\ref{eq:NonContextualityConstraint1Sup}) is satisfied. We move on to showing that Noncontextuality Constraint (\ref{eq:NonContextualityConstraint2}) is satisfied. We calculate its LHS for this ansatz:
\begin{equation}
    \begin{split}
        \sum_{x\in\{\pm1\}}  \OnticXWeakMeasureAProb{j}(x,\lambda'|\lambda) &= \sum_{x\in\{\pm1\}}\left(\delta_{\lambda\lambda'}\left[\frac{1}{2}(1-\ProbabilityOfDisturbance) + \frac{1}{2}x \ProbabilityOfMeasurement (2\Re w_{j,\lambda'}-1)\right] + \frac{1}{2}\ProbabilityOfDisturbance \UndetermindedStochasticMatrix{j}(\lambda'|\lambda)\right)\\
        &= (1-\ProbabilityOfDisturbance)\delta_{\lambda\lambda'} + \ProbabilityOfDisturbance \UndetermindedStochasticMatrix{j}(\lambda'|\lambda).
    \end{split}
\end{equation}

Comparing this result to the right-hand side of Eq. (\ref{eq:NonContextualityConstraint2}), we find that we must set $\OnticDTransformation{j}(\lambda'|\lambda)= \UndetermindedStochasticMatrix{j}(\lambda'|\lambda)$. By construction, $\UndetermindedStochasticMatrix{j}(\lambda'|\lambda)$ is a normalized stochastic matrix with all entries in $[0,1]$. Hence this assigns $\OnticDTransformation{j}(\lambda'|\lambda)$ to a normalized probability distribution between $0$ and $1$. Hence Eq. (\ref{eq:NonContextualityConstraint2}) is satisfied as well. 

We have thus demonstrated that we have adequate solutions to the equations involving $\OnticXWeakMeasureAProb{j}$. We solve the equations involving $\OnticYWeakMeasureAProb{j}$ similarly by providing an ansatz:
\begin{equation}\label{eq:AnsatzForWYSupp}
    \OnticYWeakMeasureAProb{j}(y,\lambda'|\lambda) = \delta_{\lambda\lambda'}\left[\frac{1}{2}(1-\ProbabilityOfDisturbance) + y \ProbabilityOfMeasurement \Im w_{j,\lambda'}\right] + \frac{1}{2}\ProbabilityOfDisturbance \UndetermindedStochasticMatrix{j}(\lambda'|\lambda),
\end{equation}
\noindent where $\UndetermindedStochasticMatrix{j}$ is chosen as before. Showing that this ansatz is normalized and only takes values in $[0,1]$ is done in exactly the same way as for $\OnticXWeakMeasureAProb{j}$, and we will omit it here. First, we check Eq. (\ref{eq:CorrectnessConstraint3}) for $z = +1$:

\begin{equation}
    \begin{split}
        \sum_{\lambda,\lambda' \in \Lambda} \oprep[\rho](\lambda)    \OnticYWeakMeasureAProb{j}(y, \lambda'|\lambda) \OnticMeasureBProb{k}(+1|\lambda') &= \sum_{\lambda,\lambda' \in \mathbb{Z}_d} p_{+1}^\lambda    \OnticYWeakMeasureAProb{j}(y, \lambda'|\lambda) \delta_{\lambda' k}\\
        &= \sum_{\lambda \in \mathbb{Z}_d} p_{+1}^\lambda    \OnticYWeakMeasureAProb{j}(y, k|\lambda) \\
        &= \sum_{\lambda \in \mathbb{Z}_d} \left( p_{+1}^\lambda    \delta_{\lambda k}\left[\frac{1}{2}(1-\ProbabilityOfDisturbance) + y \ProbabilityOfMeasurement \Im w_{j,k}\right] + \frac{1}{2}\ProbabilityOfDisturbance \UndetermindedStochasticMatrix{j}(k|\lambda)\right)\\
        &= \frac{1}{2}(1-\ProbabilityOfDisturbance)p_{+1}^k + y \ProbabilityOfMeasurement p_{+1}^k \Im w_{j,k} + \frac{1}{2}\ProbabilityOfDisturbance \sum_{\lambda \in \mathbb{Z}_d} p_{+1}^\lambda \UndetermindedStochasticMatrix{j}(k|\lambda)\\
        &= \frac{1}{2}(1-\ProbabilityOfDisturbance)p_{+1}^k + y \ProbabilityOfMeasurement p_{+1}^k \Im w_{j,k} + \frac{1}{2}\ProbabilityOfDisturbance \ProbOfBOnDExotD{j,k}{+1}\\
        &= \OutcomePMF{3}{j,k}(y,+1).
    \end{split}
\end{equation}

Hence Eq. (\ref{eq:CorrectnessConstraint3}) is satisfied for $z = +1$. We also check Eq. (\ref{eq:NonContextualityConstraint3}) for $z = {-1}$:
\begin{equation}
    \begin{split}
        \sum_{\lambda,\lambda' \in \Lambda} \oprep[\rho](\lambda)    \OnticYWeakMeasureAProb{j}(y, \lambda'|\lambda) \OnticMeasureBProb{k}({-1}|\lambda') &= \sum_{\lambda,\lambda' \in \mathbb{Z}_d} p_{+1}^\lambda    \OnticYWeakMeasureAProb{j}(y, \lambda'|\lambda) (1-\delta_{\lambda' k})\\
        &= \sum_{\lambda \in \mathbb{Z}_d,\lambda' \in \mathbb{Z}_d\backslash\{k\}} p_{+1}^\lambda    \OnticYWeakMeasureAProb{j}(y, \lambda'|\lambda) \\
        &= \sum_{\lambda' \in \mathbb{Z}_d\backslash\{k\}} \left(\frac{1}{2}(1-\ProbabilityOfDisturbance)p_{+1}^{\lambda'}  + y \ProbabilityOfMeasurement p_{+1}^{\lambda'}\Im w_{j,\lambda'} + \frac{1}{2}\ProbabilityOfDisturbance \ProbOfBOnDExotD{j\lambda'}{+1}\right)\\
        &= \frac{1}{2}(1-\ProbabilityOfDisturbance)p_{-1}^{k}  + y \ProbabilityOfMeasurement p_{-1}^{k} \frac{\sum_{\lambda' \in \mathbb{Z}_d\backslash\{k\}} p_{+1}^{\lambda'}\Im w_{j,\lambda'}}{p_{-1}^{k}} + \frac{1}{2}\ProbabilityOfDisturbance \ProbOfBOnDExotD{j,k}{{-1}}\\
        &= \frac{1}{2}(1-\ProbabilityOfDisturbance)p_{-1}^{k}  + y \ProbabilityOfMeasurement p_{-1}^{k} g_{-1}(\Im w_{j,\lambda'}) + \frac{1}{2}\ProbabilityOfDisturbance \ProbOfBOnDExotD{j,k}{{-1}}\\
        &= \OutcomePMF{3}{j,k}(y,{-1}).
    \end{split}
\end{equation}

Hence Eq. (\ref{eq:NonContextualityConstraint3}) is also satisfied for $z = {-1}$. Last, we will need to check the two Nontextuality Constraints associated with $\OnticYWeakMeasureAProb{j}$. We start with Eq. (\ref{eq:NonContextualityConstraint3}). 
\begin{equation}
    \begin{split}
        \sum_{\lambda'\in\Lambda} \OnticYWeakMeasureAProb{j}(y,\lambda'|\lambda) &= \sum_{\lambda'\in\mathbb{Z}_d} \left(\delta_{\lambda\lambda'}\left[\frac{1}{2}(1-\ProbabilityOfDisturbance) + y \ProbabilityOfMeasurement \Im w_{j,\lambda'}\right] + \frac{1}{2}\ProbabilityOfDisturbance \UndetermindedStochasticMatrix{j}(\lambda'|\lambda)\right)\\
        &= \sum_{\lambda'\in\mathbb{Z}_d} \delta_{\lambda\lambda'}\frac{1}{2}(1-\ProbabilityOfDisturbance)  + \frac{1}{2}\ProbabilityOfDisturbance \UndetermindedStochasticMatrix{j}(\lambda'|\lambda)\\
        &= \frac{1}{2}(1-\ProbabilityOfDisturbance)  + \frac{1}{2}\ProbabilityOfDisturbance  = \frac{1}{2},\\
    \end{split}
\end{equation}
\noindent where to go from the first to the second line we used that by Eq. \eqref{eq:TheWeakValues} the fact that $\rho$ is KD-positive implies that $\Im w_{j,\lambda'} = 0$. Hence Eq. (\ref{eq:NonContextualityConstraint3}) is satisfied. Last, we need to check Eq. (\ref{eq:NonContextualityConstraint4Sup}). 
\begin{equation}
    \begin{split}
        \sum_{y\in\{\pm1\}}  \OnticYWeakMeasureAProb{j}(y,\lambda'|\lambda) &= \sum_{y\in\{\pm1\} }\left(\delta_{\lambda\lambda'}\left[\frac{1}{2}(1-\ProbabilityOfDisturbance) + y \ProbabilityOfMeasurement \Im w_{j,\lambda'}\right] + \frac{1}{2}\ProbabilityOfDisturbance \UndetermindedStochasticMatrix{j}(\lambda'|\lambda)\right)\\
        &= (1-\ProbabilityOfDisturbance)\delta_{\lambda\lambda'}\ + \ProbabilityOfDisturbance \UndetermindedStochasticMatrix{j}(\lambda'|\lambda)\\
        &= (1-\ProbabilityOfDisturbance)\delta_{\lambda\lambda'}\ + \ProbabilityOfDisturbance \OnticDTransformation{j}(\lambda'|\lambda)\\
    \end{split}
\end{equation}
\noindent where we used that we chose $\UndetermindedStochasticMatrix{j}$ earlier to be $\OnticDTransformation{j}$. Hence we have explicitly demonstrated that if the state is KD-positive and $\epsilon$ is sufficiently small, that then there exists a noncontextual hidden variable model for Protocols 1 to 6.

\section{Bounding the Negativity of Decompositions of Exotic States}\label{app:BoundingTheNegativtyOfDecompositionsOfExoticStates}
The main result of the article depends on the claim that every decomposition of an exotic state into pure states has at least one pure state with negativity bounded away from zero. We prove this claim in this Note. More precisely, we prove the following statement:

\begin{Lemma}[Every decomposition of an exotic state has a state with bounded negativity]\label{lem:everydecomphasboundednegativity} Let $\exot \in \KDexot$ be an exotic state. There exists a $\delta > 0$ such that for every decomposition of $\exot$ into pure states, $\exot = \sum_i p_i \psi_i$, there exists at least one pure state $\psi_- \in (\psi_i)_i$ such that $\mcl{N}(\psi_-) > \delta$.
\end{Lemma}

To prove this, we aim to construct a set containing all the pure states with negativity bounded away from zero. We do this as follows. We have that $\conv{\EcalKDCpu}$ and $\{\exot\}$ are both closed and compact sets. The hyperplane separation theorem implies that there exists a Hermitian operator $\ArbHermOperator$ and real constants $d > c$ such that (i) $\Tr(\ArbHermOperator\exot) = d$ and (ii) $\forall \rho \in \conv{\EcalKDCpu}$, $\Tr(\ArbHermOperator\rho)\leq c$. We now introduce the set
\begin{equation*}
    \mathcal{A} \coloneqq \{\rho \in \mathcal{D}(\mathbb{C}^d)\;|\;\rho \text{ is pure and }\Tr(\ArbHermOperator\rho) \geq d\},
\end{equation*}
\noindent where $\mathcal{D}(\mathbb{C}^d)$ denotes the set of density operators on $\mathbb{C}^d$, and we prove that this set has the following desirable properties.

\begin{Lemma}[Properties of $\mathcal{A}$]\label{lem:PropertiesOfA} The set $\mathcal{A}$ has the following properties
    \begin{enumerate}[label = (\roman*)]
        \item For every decomposition of an exotic state $\exot$ into pure states, $\exot = \sum_i p_i \psi_i$, there exists at least one pure state $\psi_i$ such that $\psi_i \in \mathcal{A}$.
        \item There exists a real constant $\delta>0$ such that for all $\rho \in \mathcal{A}$, $\mathcal{N}(\rho) \geq \delta$.
    \end{enumerate}
\end{Lemma}

\begin{proof}
    (i) by contradiction. Suppose that for all $i$, $\psi_i \notin \mathcal{A}$. We then calculate
    \begin{equation}
        \Tr(\ArbHermOperator\exot) = \sum_i p_i\Tr(\ArbHermOperator\psi_i) < \sum_i p_i d = d ,
    \end{equation}
    \noindent where we used the defining property of $\mathcal{A}$ in the inequality. However, from the definition of $\ArbHermOperator$, we know that $\Tr(\ArbHermOperator\exot) = d$. We thus find a contradiction, establishing the result.

    (ii) $\mathcal{A}$ is closed and compact, and  $\mathcal{N}$ is continuous. The Extreme Value Theorem then implies that $\mathcal{N}$ is bounded on $\mathcal{A}$, and that $\mathcal{N}$ attains this bound at some $\rho_{\text{min}}\in \mathcal{A}$.

    Now suppose that $\mathcal{N}(\rho_{\text{min}}) = 0$. This implies that $\rho_{\text{min}}$ is KD-positive. All states in $\mathcal{A}$ are pure, so we have that $\rho_{\text{min}} \in \EcalKDCpu$. From the definition of $\ArbHermOperator$, we this yields 
    \begin{equation}
        \Tr(\ArbHermOperator \rho_{\text{min}})  \leq c < d.
    \end{equation}

    However, since $\rho_{\text{min}} \in \mathcal{A}$, we have that $\Tr(\ArbHermOperator \rho_{\text{min}}) \geq d$ as well. We thus find a contradiction, forcing us to conclude that $\mathcal{N}(\rho_{\text{min}}) \neq 0$. We may thus set $\delta = \mathcal{N}(\rho_{\text{min}}) > 0$ to obtain the desired result.
\end{proof} 

The proof of Lemma \ref{lem:everydecomphasboundednegativity} then straightforwardly follows from Lemma \ref{lem:PropertiesOfA}.

\end{document}